\documentclass{amsart}
\usepackage{mathrsfs}
\usepackage{amsfonts}
\usepackage{color}
\usepackage{framed}
\usepackage{amssymb}
\usepackage{amscd}
\usepackage{amsmath}

\usepackage{indentfirst,latexsym,bm,amsmath}
\usepackage{graphics}

\def\mtline#1{\hbox to#1{\hrulefill}}

\newtheorem{theorem}{Theorem}

\newtheorem{remark}{Remark}

\newtheorem{definition}{Definition}

\newtheorem{conjecture}{Conjecture}

\newtheorem{proposition}{Proposition}

\renewcommand{\thetable}{\arabic{section}.\arabic{table}}

\usepackage{amsmath}
\usepackage{amssymb}
\usepackage{amsfonts}
\usepackage{amsmath,amsfonts,amssymb}
\usepackage{fancyhdr}
\usepackage[top=1in, bottom=1in, left=1.25in, right=1.25in]{geometry}
\begin{document}
\fancyhf{}

\fancyhead[LO]{Classification of Toric Codes} \cfoot{\thepage}

\title{\Large{\textbf{ On Classification of Toric Surface Codes
of low Dimension}}
\date{}
 }
 \author{Xue Luo}
\thanks{School of Mathematics and System Sciences, Beihang University, Beijing, P. R. China, 100191, email:\,xluo@buaa.edu.cn}
\author{Stephen S.-T. Yau}
\thanks{Department of Mathematical Sciences, Tsinghua University, Beijing, P. R. China, 100084, and Department of Mathematics, Statistics, and Computer Science, University of Illinois at Chicago, IL, 60607, email:\,yau@uic.edu}
\author{Mingyi Zhang}
\thanks{Department of Mathematical Sciences, Tsinghua University, Beijing, P. R. of China, 100084, email:\,zmysljk@gmail.com}
\author{Huaiqing Zuo}
\thanks{Mathematical Sciences Center, Tsinghua University,  Beijing, P. R. China,  100084, email:\,hqzuo@math.tsinghua.edu.cn}
 \maketitle

 \renewcommand{\abstractname}{Abstract}
\begin{abstract}
This work is a natural continuation of our previous work \cite{yz}.
In this
paper, we give a complete classification of toric surface
codes of dimension equal to 6, except a special pair,
$C_{P_6^{(4)}}$ and $C_{P_6^{(5)}}$ over $\mathbb{F}_8$. Also, we
give an example, $C_{P_6^{(5)}}$ and $C_{P_6^{(6)}}$ over
$\mathbb{F}_7$, to illustrate that two monomially equivalent toric
codes can be constructed from two lattice non-equivalent polygons.
\end{abstract}

\begin{center}
\section{Introduction}
\end{center}

Toric codes, which were introduced by J. Hansen \cite{h1}, are
constructed on toric varieties. These toric codes have attracted quite a bit of attentions in the last decade, because they are, in some sense, a natural
extension of Reed-Solomon codes, which have been studied recently in
\cite{dgv}, \cite{h1}, \cite{h2}, \cite{ls1}, \cite{ls2}, \cite{j},
etc. 

Compared to the other codes, the toric codes have their own
advantage for study. The properties of these codes are closely tied
to the geometry of the toric surface $X_P$ associated with the
normal fan $\triangle_P$ of the polygon $P$. Thanks to this
advantage, D. Ruano \cite{rua} estimated the minimum distance using
intersection theory and mixed volumes, extending the methods of J.
Hansen for plane polygons. J. Little and H. Schenck \cite{ls1}
obtained upper and lower bounds on the minimum distance of a toric
code constructed from a polygon $P\subset\mathbb{R}^2$ by examining
Minkowski sum decompositions of subpolygons of $P$. The most
interesting things are that J. Little and R. Schwarz \cite{ls2} used
a more elementary approach to determine the minimum distance of
toric codes from simplices and rectangular polytopes.  They also proved a general result
that if there is a unimodular integer affine transformation taking
one polytope $P_1$ to another polytope $P_2$, i.e. $P_1$ and $P_2$ are
lattice equivalent (Definition \ref{df3}), then the corresponding
toric codes are monomially equivalent (hence have the same
parameters). However, the reverse implication is not true. An
explicit example will be given in our paper to illustrate this statement.
Based on this useful tool, they classified the toric surface codes
with a small dimension. However, one case of toric codes of dimension
$5$ was missing in their classification of toric surface codes. In
\cite{yz}, the second and the last author of this paper supplemented the missing case and completed the proof of classification of toric codes with
dimension less than or equal to $5$.  There are other families of higher dimensional toric codes for which the minimum distance is computed explicitly, see \cite{ss2}.

In this paper,
we give a complete classification of toric surface codes of
dimension equal to $6$. Also some interesting phenomena
have been discovered in the process of our proofs. On the one hand, we give an
explicit example that two monomially equivalent toric codes can be constructed from two lattice non-equivalent polygons, see Proposition
\ref{over F7}. On the other hand, the number of the codewords in $C_P$ over
$\mathbb{F}_q$ with some particular weight can be varied by the
choice of $q$, see Table \ref{table-2}-\ref{table-5}. The methods in this paper shed a light on classification of toric surface codes of higher dimension and it may give better champion codes than those in \cite{bk}.

The main results in this paper are stated below:
\begin{theorem}\label{th1}
 Every toric surface code
with $3\leq k\leq6$,  where $k$ is the dimension of the code, is
monomially equivalent to one constructed from one of the
polygons in Fig. \ref{fig-k=3}-\ref{fig-k=5} or Fig. \ref{fig-k=6}.
\end{theorem}
\bigskip

\begin{figure}[!h]
\setlength{\unitlength}{1mm}
\begin{picture}(50,30)(-15,-8)
\put(0,0){\vector(0,1){25}} \put(0,0){\vector(1,0){40}}
\multiput(0,0)(15,0){3}{\circle*{1.5}} \put(-0.5,15){\line(1,0){1}}
\thicklines \put(0,0){\line(1,0){30}} \put(1,-6){1} \put(15,-6){$x$}
\put(30,-6){$x^2$} \put(-6,15){$y$} \put(10,20){$P_3^{(1)}$}
\end{picture}
\hspace{10mm} \setlength{\unitlength}{1mm}
\begin{picture}(50,30)(-10,-8)
\put(0,0){\vector(0,1){25}} \put(0,0){\vector(1,0){40}} \thicklines
\put(0,15){\line(1,-1){15}} \put(0,15){\line(0,-1){15}}
\put(0,0){\line(1,0){15}} \multiput(0,0)(15,0){2}{\circle*{1.5}}
\put(0,15){\circle*{1.5}} \thinlines \put(30,-0.5){\line(0,1){1}}
 \put(1,-6){1} \put(15,-6){$x$}
\put(30,-6){$x^2$} \put(-6,15){$y$} \put(10,20){$P_3^{(2)}$}
\end{picture}
	\caption{Polygons yielding toric codes with $k=3$}\label{fig-k=3}
\end{figure}
\begin{figure}[!h]
\setlength{\unitlength}{1mm}
\begin{picture}(40,30)(-12,-12)
\put(0,0){\vector(0,1){18}} \put(0,0){\vector(1,0){40}}
\multiput(0,0)(12,0){4}{\circle*{1.5}} \put(-0.5,12){\line(1,0){1}}
\thicklines \put(0,0){\line(1,0){36}} \put(1,-4){1} \put(13,-4){$x$}
\put(25,-4){$x^2$} \put(37,-4){$x^3$} \put(-4,12){$y$}
\put(10,20){$P_4^{(1)}$}
\end{picture}
\hspace{20mm} \setlength{\unitlength}{1mm}
\begin{picture}(40,30)(-12,-12)
\put(0,0){\vector(0,1){18}} \put(0,0){\vector(1,0){40}}
\multiput(0,0)(12,0){3}{\circle*{1.5}} \put(0,12){\circle*{1.5}}
\put(36,-0.5){\line(0,1){1}} \thicklines \put(0,0){\line(1,0){24}}
\put(0,0){\line(0,1){12}} \put(0,12){\line(2,-1){24}} \put(1,-4){1}
\put(13,-4){$x$} \put(25,-4){$x^2$} \put(37,-4){$x^3$}
\put(-4,12){$y$} \put(10,20){$P_4^{(2)}$}
\end{picture}\\
\setlength{\unitlength}{1mm}
\begin{picture}(40,30)(-12,-12)
\put(0,0){\vector(0,1){18}} \put(0,0){\vector(1,0){30}}
\multiput(0,0)(10,0){2}{\circle*{1.5}} \put(0,10){\circle*{1.5}}
\put(10,10){\circle*{1.5}} \put(20,-0.5){\line(0,1){1}} \thicklines
\put(0,0){\line(1,0){10}} \put(0,0){\line(0,1){10}}
\put(10,10){\line(-1,0){10}} \put(10,10){\line(0,-1){10}}
\put(1,-4){1} \put(10,-4){$x$} \put(20,-4){$x^2$}  \put(-4,10){$y$}
\put(10,20){$P_4^{(3)}$}
\end{picture}
\hspace{20mm} \setlength{\unitlength}{1mm}
\begin{picture}(40,30)(-12,-12)
\put(0,-12){\vector(0,1){30}} \put(-12,0){\vector(1,0){42}}
\multiput(0,0)(10,0){2}{\circle*{1.5}} \put(0,10){\circle*{1.5}}
\put(-10,-10){\circle*{1.5}} \put(-0.5,-10){\line(1,0){1}}
\put(20,-0.5){\line(0,1){1}} \put(-10,-0.5){\line(0,1){1}}
\thicklines \put(-10,-10){\line(2,1){20}}
\put(-10,-10){\line(1,2){10}} \put(0,10){\line(1,-1){10}}
\put(1,-4){1} \put(10,-4){$x$} \put(20,-4){$x^2$}  \put(-4,10){$y$}
\put(-10,1){$x^{-1}$} \put(1,-10){$y^{-1}$} \put(10,10){$P_4^{(4)}$}
\end{picture}	
	\caption{Polygons yielding toric codes with $k=4$}\label{fig-k=4}
\end{figure}

\begin{figure}[!h]
\setlength{\unitlength}{1mm}
\begin{picture}(30,30)(-8,-8)
\put(0,0){\vector(0,1){20}} \put(0,0){\vector(1,0){32}}
\multiput(0,0)(7,0){5}{\circle*{1.5}} \put(-0.5,7){\line(1,0){1}}
\put(-0.5,14){\line(1,0){1}} \thicklines \put(0,0){\line(1,0){28}}
\put(0,-4){1} \put(7,-4){$x$} \put(14,-4){$x^2$} \put(21,-4){$x^3$}
\put(28,-4){$x^4$} \put(-4,7){$y$} \put(-4,14){$y^2$}
\put(9,17){$P_5^{(1)}$}
\end{picture}
\hspace{10mm} \setlength{\unitlength}{1mm}
\begin{picture}(30,30)(-8,-8)
\put(0,0){\vector(0,1){20}} \put(0,0){\vector(1,0){30}}
\multiput(0,0)(7,0){4}{\circle*{1.5}} \put(0,7){\circle*{1.5}}
\put(-0.5,14){\line(1,0){1}} \thicklines \put(0,0){\line(1,0){21}}
\put(0,0){\line(0,1){7}} \put(0,7){\line(3,-1){21}}  \put(0,-4){1}
\put(7,-4){$x$} \put(14,-4){$x^2$} \put(21,-4){$x^3$}
 \put(-4,7){$y$} \put(-4,14){$y^2$}
\put(9,17){$P_5^{(2)}$}
\end{picture}
\hspace{10mm}
\begin{picture}(30,30)(-8,-8)
\put(0,0){\vector(0,1){20}} \put(0,0){\vector(1,0){30}}
\multiput(0,0)(7,0){3}{\circle*{1.5}} \put(0,7){\circle*{1.5}}
\put(7,7){\circle*{1.5}} \put(-0.5,14){\line(1,0){1}}
\put(21,-0.5){\line(0,1){1}} \thicklines \put(0,0){\line(1,0){14}}
\put(0,0){\line(0,1){7}} \put(0,7){\line(1,0){7}}
\put(7,7){\line(1,-1){7}} \put(0,-4){1} \put(7,-4){$x$}
\put(14,-4){$x^2$} \put(21,-4){$x^3$}
 \put(-4,7){$y$} \put(-4,14){$y^2$}
\put(9,17){$P_5^{(3)}$}
\end{picture}\\[20pt]
\setlength{\unitlength}{1mm}
\begin{picture}(30,30)(-8,-8)
\put(0,0){\vector(0,1){24}} \put(0,0){\vector(1,0){30}}
\multiput(0,0)(7,0){2}{\circle*{1.5}}
\multiput(0,7)(7,0){2}{\circle*{1.5}} \put(14,14){\circle*{1.5}}
\put(-0.5,14){\line(1,0){1}} \put(-0.5,21){\line(1,0){1}}
\put(14,-0.5){\line(0,1){1}} \put(21,-0.5){\line(0,1){1}}
\thicklines \put(0,0){\line(1,0){7}} \put(0,0){\line(0,1){7}}
\put(0,7){\line(2,1){14}} \put(7,0){\line(1,2){7}} \put(0,-4){1}
\put(7,-4){$x$} \put(14,-4){$x^2$} \put(21,-4){$x^3$}
 \put(-4,7){$y$} \put(-4,14){$y^2$} \put(-4,21){$y^3$}
\put(9,21){$P_5^{(4)}$}
\end{picture}
\hspace{10mm} \setlength{\unitlength}{1mm}
\begin{picture}(30,30)(-8,-8)
\put(0,-8){\vector(0,1){32}} \put(-8,0){\vector(1,0){38}}
\multiput(0,0)(7,0){3}{\circle*{1.5}}
\multiput(0,-7)(0,14){2}{\circle*{1.5}} \put(-0.5,14){\line(1,0){1}}
\put(-0.5,21){\line(1,0){1}} \put(21,-0.5){\line(0,1){1}}
\put(-7,-0.5){\line(0,1){1}} \thicklines \put(0,-7){\line(0,1){14}}
\put(0,-7){\line(2,1){14}} \put(0,7){\line(2,-1){14}} \put(-2,-4){1}
\put(5,-4){$x$} \put(12,-4){$x^2$} \put(19,-4){$x^3$}
\put(-8,1){$x^{-1}$}
 \put(-4,7){$y$} \put(-4,14){$y^2$} \put(-4,21){$y^3$} \put(-6,-9){$y^{-1}$}
\put(9,21){$P_5^{(5)}$}
\end{picture}
\hspace{10mm} \setlength{\unitlength}{1mm}
\begin{picture}(30,30)(-8,-8)
\put(0,-8){\vector(0,1){32}} \put(-8,0){\vector(1,0){38}}
\multiput(0,0)(7,0){2}{\circle*{1.5}}
\multiput(0,7)(0,7){2}{\circle*{1.5}} \put(-7,-7){\circle*{1.5}}
\put(-0.5,21){\line(1,0){1}} \put(-0.5,-7){\line(1,0){1}}
\put(21,-0.5){\line(0,1){1}} \put(14,-0.5){\line(0,1){1}}
\put(-7,-0.5){\line(0,1){1}} \thicklines \put(-7,-7){\line(1,3){7}}
\put(-7,-7){\line(2,1){14}} \put(0,14){\line(1,-2){7}}
\put(-2,-4){1} \put(5,-4){$x$} \put(12,-4){$x^2$} \put(19,-4){$x^3$}
\put(-9,1){$x^{-1}$}
 \put(-4,7){$y$} \put(-4,14){$y^2$} \put(-4,21){$y^3$} \put(0,-9){$y^{-1}$}
\put(9,21){$P_5^{(6)}$}
\end{picture}\\[20pt]
\setlength{\unitlength}{1mm}
\begin{picture}(30,30)(-8,-8)
\put(0,-10){\vector(0,1){34}} \put(-10,0){\vector(1,0){40}}
\multiput(-8,0)(8,0){3}{\circle*{1.5}}
\multiput(0,-8)(0,16){2}{\circle*{1.5}} \put(16,-0.5){\line(0,1){1}}
\put(-0.5,16){\line(1,0){1}} \put(24,-0.5){\line(0,1){1}}
\thicklines \put(-8,0){\line(1,-1){8}} \put(-8,0){\line(1,1){8}}
\put(0,8){\line(1,-1){8}} \put(0,-8){\line(1,1){8}} \put(1,-4){1}
\put(8,-4){$x$} \put(16,-4){$x^2$} \put(24,-4){$x^3$}
\put(-10,-4){$x^{-1}$}
 \put(-4,8){$y$} \put(-4,16){$y^2$} \put(0,-10){$y^{-1}$}
\put(9,17){$P_5^{(7)}$}
\end{picture}
	\caption{Polygons yielding toric codes with $k=5$}\label{fig-k=5}
\end{figure}

\begin{figure}
\hspace{-10mm}\setlength{\unitlength}{1mm}
\begin{picture}(45,30)(-8,-8)
\put(0,-1){\vector(0,1){22}} \put(-1,0){\vector(1,0){44}}
\multiput(0,0)(8,0){6}{\circle*{1.5}}
\multiput(-0.5,0)(0,8){3}{\line(1,0){1}} \thicklines
\put(0,0){\line(1,0){40}}  \put(0,-5){1} \put(7,-5){$x$}
\put(14,-5){$x^2$} \put(23,-5){$x^3$} \put(30,-5){$x^4$}
\put(37,-5){$x^5$}  \put(-5,8){$y$} \put(-5,16){$y^2$}
\put(15,18){$P_6^{(1)}$}
\end{picture}
\hspace{5mm}
\begin{picture}(35,30)(-8,-8)
\put(0,-1){\vector(0,1){22}} \put(-1,0){\vector(1,0){35}}
\multiput(0,0)(8,0){5}{\circle*{1.5}} \put(0,8){\circle*{1.5}}
\multiput(-0.5,0)(0,8){3}{\line(1,0){1}} \thicklines
\put(0,0){\line(1,0){32}} \thicklines \put(0,0){\line(0,1){8}}
\thicklines \put(0,8){\line(4,-1){32}} \put(0,-5){1} \put(7,-5){$x$}
\put(14,-5){$x^2$} \put(22,-5){$x^3$} \put(30,-5){$x^4$}
\put(-5,8){$y$} \put(-5,16){$y^2$} \put(15,18){$P_6^{(2)}$}
\end{picture}
\hspace{5mm}
\begin{picture}(32,30)(-8,-8)
\put(0,-1){\vector(0,1){22}} \put(-1,0){\vector(1,0){32}}
\multiput(0,0)(8,0){4}{\circle*{1.5}}
\multiput(0,8)(8,0){2}{\circle*{1.5}}
\multiput(-0.5,0)(0,8){3}{\line(1,0){1}} \thicklines
\put(0,0){\line(1,0){24}} \thicklines \put(0,0){\line(0,1){8}}
\thicklines \put(0,8){\line(1,0){8}} \thicklines
\put(8,8){\line(2,-1){16}} \put(0,-5){1} \put(7,-5){$x$}
\put(14,-5){$x^2$} \put(22,-5){$x^3$}
 \put(-5,8){$y$} \put(-5,16){$y^2$}
\put(15,18){$P_6^{(3)}$}
\end{picture}\\[10pt]
\setlength{\unitlength}{1mm}
\begin{picture}(35,30)(-8,-8)
\put(0,-9){\vector(0,1){30}} \put(-1,0){\vector(1,0){32}}
\multiput(0,0)(8,0){4}{\circle*{1.5}}
\multiput(0,-8)(0,16){2}{\circle*{1.5}} \put(-0.5,16){\line(1,0){1}}
\thicklines \put(0,-8){\line(3,1){24}} \thicklines
\put(0,-8){\line(0,1){16}} \thicklines \put(0,8){\line(3,-1){24}}
\put(1,-4){1} \put(8,-4){$x$} \put(16,-4){$x^2$} \put(24,-4){$x^3$}
\put(-5,8){$y$} \put(-5,16){$y^2$} \put(-6,-8){$y^{-1}$}
\put(15,18){$P_6^{(4)}$}
\end{picture}
\hspace{30pt}
\begin{picture}(40,30)(-8,-8)
\put(0,-9){\vector(0,1){30}} \put(-9,0){\vector(1,0){40}}
\multiput(0,0)(8,0){4}{\circle*{1.5}} \put(0,8){\circle*{1.5}}
\put(-8,-8){\circle*{1.5}} \put(-8,-0.5){\line(0,1){1}}
\put(-0.5,16){\line(1,0){1}} \put(-0.5,-8){\line(1,0){1}}
\thicklines \put(-8,-8){\line(4,1){32}} \thicklines
\put(0,8){\line(-1,-2){8}} \thicklines \put(0,8){\line(3,-1){24}}
\put(0,-5){1} \put(7,-5){$x$} \put(14,-5){$x^2$} \put(22,-5){$x^3$}
\put(-10,-5){$x^{-1}$}
 \put(-5,8){$y$} \put(-5,16){$y^2$} \put(-6,-9){$y^{-1}$}
\put(15,18){$P_6^{(5)}$}
\end{picture}
\hspace{10mm}
\begin{picture}(30,30)(-8,-8)
\put(-8,-9){\vector(0,1){30}} \put(-9,-8){\vector(1,0){30}}
\multiput(-8,-8)(8,0){2}{\circle*{1.5}}
\multiput(-8,0)(8,0){2}{\circle*{1.5}} \put(8,8){\circle*{1.5}}
\put(16,16){\circle*{1.5}} \multiput(8,-8.5)(8,0){2}{\line(0,1){1}}
\multiput(-8.5,8)(0,8){2}{\line(1,0){1}} \thicklines
\put(-8,-8){\line(1,0){8}} \thicklines \put(-8,-8){\line(0,1){8}}
\thicklines \put(-8,0){\line(3,2){24}} \thicklines
\put(0,-8){\line(2,3){16}} \put(-8,-13){1} \put(-1,-13){$x$}
\put(6,-13){$x^2$} \put(14,-13){$x^3$} \put(-13,0){$y$}
\put(-13,8){$y^2$} \put(-13,16){$y^3$} \put(10,18){$P_6^{(6)}$}
\end{picture}\\[20pt]
\hspace{-10mm}\setlength{\unitlength}{1mm}
\begin{picture}(30,30)(-8,-8)
\put(0,-10){\vector(0,1){25}} \put(-1,0){\vector(1,0){30}}
\multiput(0,0)(8,0){4}{\circle*{1.5}} \put(0,8) {\circle*{1.5}}
\put(24,-8){\circle*{1.5}} \put(-0.5,-8){\line(1,0){1}} \thicklines
\put(0,0){\line(0,1){8}} \thicklines \put(0,8){\line(3,-1){24}}
\thicklines \put(0,0){\line(3,-1){24}} \thicklines
\put(24,0){\line(0,-1){8}} \put(0,-5){1} \put(7,-5){$x$}
\put(14,-5){$x^2$} \put(22,-5){$x^3$} \put(-5,8){$y$}
\put(-6,-8){$y^{-1}$} \put(11,15){$P_6^{(7)}$}
\end{picture}
\hspace{20mm}
\begin{picture}(30,30)(-8,-8)
\put(0,-9){\vector(0,1){25}} \put(-9,0){\vector(1,0){30}}
\multiput(-8,0)(8,0){4}{\circle*{1.5}} \put(0,8){\circle*{1.5}}
\put(0,-8){\circle*{1.5}} \thicklines \put(-8,0){\line(1,1){8}}
\thicklines \put(-8,0){\line(1,-1){8}} \thicklines
\put(0,8){\line(2,-1){16}} \thicklines \put(0,-8){\line(2,1){16}}
\put(0,-5){1} \put(7,-5){$x$} \put(14,-5){$x^2$}
\put(-9,-5){$x^{-1}$}
 \put(-5,8){$y$} \put(-5,16){$y^2$} \put(-6,-8){$y^{-1}$}
\put(11,15){$P_6^{(8)}$}
\end{picture}
\hspace{10mm}
\begin{picture}(25,30)(-8,-8)
\put(0,-9){\vector(0,1){25}} \put(-1,0){\vector(1,0){25}}
\multiput(0,0)(8,0){3}{\circle*{1.5}}
\multiput(0,8)(8,0){2}{\circle*{1.5}} \put(0,-8){\circle*{1.5}}
\thicklines \put(0,-8){\line(0,1){16}} \thicklines
\put(0,8){\line(1,0){8}} \thicklines \put(8,8){\line(1,-1){8}}
\thicklines \put(16,0){\line(-2,-1){16}} \put(0,-5){1}
\put(7,-5){$x$} \put(14,-5){$x^2$}
 \put(-6,8){$y$} \put(-6,-8){$y^{-1}$}
\put(11,15){$P_6^{(9)}$}
\end{picture}\\[20pt]
\hspace{-1mm}\setlength{\unitlength}{1mm}
\begin{picture}(30,40)(-8,-8)
\put(0,-9){\vector(0,1){35}} \put(-9,0){\vector(1,0){30}}
\multiput(0,0)(8,0){3}{\circle*{1.5}} \put(0,-8){\circle*{1.5}}
\put(-8,16){\circle*{1.5}} \put(0,8){\circle*{1.5}}
\put(-0.5,16){\line(1,0){1}} \put(-8,-0.5){\line(0,1){1}}
\thicklines \put(0,-8){\line(-1,3){8}} \thicklines
\put(-8,16){\line(3,-2){24}} \thicklines
\put(16,0){\line(-2,-1){16}} \put(0,-5){1} \put(7,-5){$x$}
\put(14,-5){$x^2$} \put(-9,-5){$x^{-1}$} \put(-5,8){$y$}
\put(-6,-9){$y^{-1}$} \put(-5,16){$y^2$} \put(-5,24){$y^3$}
\put(13,18){$P_6^{(10)}$}
\end{picture}
\hspace{10mm}
\begin{picture}(25,40)(-8,-8)
\put(0,-9){\vector(0,1){35}} \put(-9,0){\vector(1,0){25}}
\multiput(0,0)(0,8){3}{\circle*{1.5}} \put(-8,-8){\circle*{1.5}}
\multiput(8,0)(0,8){2}{\circle*{1.5}} \put(-0.5,-8){\line(1,0){1}}
\put(-8,-0.5){\line(0,1){1}} \thicklines \put(-8,-8){\line(2,1){16}}
\thicklines \put(8,0){\line(0,1){8}} \thicklines
\put(8,8){\line(-1,1){8}} \thicklines \put(0,16){\line(-1,-3){8}}
\put(0,-5){1} \put(7,-5){$x$} \put(-9,-5){$x^{-1}$} \put(-5,8){$y$}
\put(-6,-9){$y^{-1}$} \put(-5,16){$y^2$} \put(-5,24){$y^3$}
\put(15,18){$P_6^{(11)}$}
\end{picture}
\hspace{10mm}
\begin{picture}(25,40)(-8,-8)
\put(0,-9){\vector(0,1){30}} \put(-9,0){\vector(1,0){25}}
\multiput(-8,0)(8,0){3}{\circle*{1.5}} \put(0,-8){\circle*{1.5}}
\multiput(0,8)(8,0){2}{\circle*{1.5}} \thicklines
\put(-8,0){\line(1,-1){8}} \put(-0.5,16){\line(1,0){1}}\thicklines
\put(-8,0){\line(1,1){8}} \thicklines \put(0,8){\line(1,0){8}}
\thicklines \put(8,8){\line(0,-1){8}} \thicklines
\put(8,0){\line(-1,-1){8}} \put(0,-5){1} \put(7,-5){$x$}
\put(-9,-5){$x^{-1}$}
 \put(-5,8){$y$} \put(-5,16){$y^2$} \put(-6,-9){$y^{-1}$}
\put(15,18){$P_6^{(12)}$}
\end{picture}\\[-10pt]
\setlength{\unitlength}{1mm}
\begin{picture}(30,50)(-8,-8)
\put(0,-1){\vector(0,1){25}} \put(-1,0){\vector(1,0){30}}
\multiput(0,0)(8,0){3}{\circle*{1.5}}
\multiput(0,8)(8,0){2}{\circle*{1.5}} \put(0,16){\circle*{1.5}}
\thicklines \put(0,0){\line(1,0){16}} \thicklines
\put(16,0){\line(-1,1){16}} \thicklines \put(0,16){\line(0,-1){16}}
\put(0,-5){1} \put(7,-5){$x$} \put(14,-5){$x^2$}
 \put(-5,8){$y$} \put(-5,16){$y^2$}
\put(15,20){$P_6^{(13)}$}
\end{picture}
\hspace{10mm}
\begin{picture}(30,40)(-8,-8)
\put(0,-1){\vector(0,1){25}} \put(-1,0){\vector(1,0){30}}
\multiput(0,0)(8,0){3}{\circle*{1.5}}
\multiput(0,8)(8,0){3}{\circle*{1.5}} \put(-0.5,16){\line(1,0){1}}
\thicklines \put(0,0){\line(1,0){16}} \thicklines
\put(0,0){\line(0,1){8}} \thicklines \put(0,8){\line(1,0){16}}
\thicklines \put(16,8){\line(0,-1){8}} \put(0,-5){1} \put(7,-5){$x$}
\put(14,-5){$x^2$}
 \put(-5,8){$y$} \put(-5,16){$y^2$}
\put(15,20){$P_6^{(14)}$}
\end{picture}
	\caption{Polygons yielding toric codes with $k=6$}\label{fig-k=6}
\end{figure}

\begin{theorem}\label{th2}$C_{P_6^{(i)}}$ and $C_{P_6^{(j)}}$ are not monomially equivalent over $\mathbb{F}_q$ for all $q\ge 7$, except that
\begin{enumerate} 
	\item $C_{P_6^{(5)}}$ and $C_{P_6^{(6)}}$ over $\mathbb{F}_7$ are monomially equivalent;
	\item the monomial equivalence of $C_{P_6^{(4)}}$ and $C_{P_6^{(5)}}$ over $\mathbb{F}_8$ remains open.
\end{enumerate}
\end{theorem}

The above theorems yield
a (almost) complete classification of the toric codes of dimension $\leq 6$  up to monomial equivalence. 
Based on the fact that the enumerator polynomial of $C_{P_6^{(4)}}$ and
$C_{P_6^{(5)}}$ over $\mathbb{F}_8$ are exactly the same (see Table \ref{table-a.1}), we conjecture the following:

\begin{conjecture}
    $C_{P_6^{(4)}}$ and $C_{P_6^{(5)}}$ over $\mathbb{F}_8$ are
monomially equivalent.
\end{conjecture}

This paper is organized as follows. In Section 2, some preliminaries have been introduced. Section 3 is devoted to the
proofs of the theorems. All the data computed by GAP (code from \cite{j}) are collected in the tables.

\begin{center}
\section{Preliminaries}
\end{center}

In this section, we shall recall some basic definitions and results
to be used later in this paper. We shall follow the terminology and
notations for toric codes in \cite{ls2}.

\subsection{Toric surface codes}
Given a finite field $\mathbb{F}_q$ where $q$ is a power of prime number. Let $P$
 be any convex lattice polygon contained in $\square_{q-1}=[0,q-2]^2$. We associate $P$ with a $\mathbb{F}_q$-vector space of polynomials spanned by the bivariate power monomials:
 
 $$\mathcal{L}(P)=\text{Span}_{\mathbb{F}_q}\{x^{m_1}y^{m_2}|(m_1,m_2)\in P\}.$$
 
 The toric surface code $C_P$ (\cite{ss}) is a linear code with codewords the strings of values of $f\in\mathcal{L}(P)$ at all points of the algebraic torus $(\mathbb{F}_q^*)^2$:
 
 $$C_P=\{(f(t),t\in(\mathbb{F}_q^*)^2)|f\in\mathcal{L}(P)\}.$$
\subsection{ Minkowski Sum and Minimum Distance of Toric Codes}

For some special polygons $P$, one can compute the minimum distance
of the toric surface code $C_P$, say the rectangles and triangles.

Let $P_{k,l}^\square={\rm conv}\{(0,0),(k,0),(0,l),(k,l)\}$ be the convex
hull of the vectors $(0,0)$, $(k,0)$, $(0,l)$, $(k,l)$. Let $\square_{q-1}=[0,q-2]^2\subset\mathbb{Z}^2$. The minimum distance of $C_{P_{k,l}^\square}$ is given in the following theorem.

\begin{theorem}\label{th4}\textup{(\cite{ls2})}  Let $k,l <
q-1$, so that $P_{k,l}^\square\subset\Box_{q-1}\subset\mathbb{R}^2$.
Then the minimum distance of the toric surface code
$C_{P_{k,l}^\square}$ is
$$
d(C_{P_{k,l}^\square})=(q-1)^2-(k+l)(q-1)+kl=((q-1)-k)((q-1)-l).
$$
\end{theorem}

Let $P_{k,l}^\vartriangle=\rm{conv}\{(0,0),(k,0),(0,l)\}$ be the convex
hull of the vectors $(0,0),(k,0),(0,l)$. Similarly, the minimum distance of $C_{P_{k,l}^\triangle}$ is given below:

\begin{theorem}\label{sjx}\textup{(\cite{ls2})}
If $P_{k,l}^\triangle\subset\Box_{q-1}\subset\mathbb{R}^2$, and
$m=\max{\{k,l\}}$, then
$$
d(C_{P_{k,l}^\triangle})=(q-1)^2-m(q-1).
$$
\end{theorem}
\begin{remark}
These two theorems above can be generalized to higher dimensional
case, see \textup{\cite{ls2}}.
\end{remark}

In the paper \cite{ss}, the authors give a good bound for the minimum distance of $C_P$ in terms of certain geometric invariant $L(P)$, the so-called full Minkowski length of $P$.

\begin{definition}\label{df1} Let $P$ and $Q$ be two subsets of
$\mathbb{R}^n$. The Minkowski sum is obtained by taking the
pointwise sum of $P$ and $Q$:
$$
P+Q=\{x+y\mid x\in P,y\in Q\}.
$$
\end{definition}

Let $P$ be a lattice polytope in $\mathbb{R}^n$. Consider a Minkowski decomposition 
 $$P=P_1+\cdots+P_l$$
 into lattice polytopes $P_i$ of positive dimension. 
 Let $l(P)$ be the largest number of summands in such decompositions of $P$, and called the Minkowski length of $P$.
 
 \begin{definition}\textup{(\cite{ss})} The full Minkowski length of $P$ is the maximum of the Minkowski lengths of all subpolytopes $Q$ in $P$,
 $$L(P):=max\{l(Q)|Q\subset P\}.$$
 \end{definition}
 
We shall use the results in \cite{ss} to give a bound of the minimum distance of $C_{P}$:
 \begin{theorem}\textup{(\cite{ss})}
 Let $P\subset\Box_{q-1}$ be a lattice polygon with area $A$ and full Minkowski length $L$. For $q\ge max(23,(c+\sqrt{c^2+5/2})^2)$, where $c=A/2-L+9/4$, the minimum distance of the toric surface code $C_P$ satisfies 
 $$d(C_P)\ge(q-1)^2-L(q-1)-2\sqrt{q}+1.$$
 \end{theorem}
 
 With the condition that no factorization $f=f_1\cdots f_{L(P)}$ for all $f\in\mathcal{L}(P)$ contains an exceptional triangle (a triangle with exactly 1 interior and 3 boundary lattice points), we have a better bound for the minimum distance of $C_P$:
 
 \begin{proposition}\label{bound}\textup{(\cite{ss})}
 Let $P\subset\Box_{q-1}$ be a lattice polygon with area $A$ and full Minkowski length $L$. Under the above condition on $P$, for $q\ge \max(37,(c+\sqrt{c^2+2})^2)$, where $c=A/2-L+11/4$, the minimum distance of the toric surface code $C_P$ satisfies 
 $$d(C_P)\ge(q-1)^2-L(q-1).$$
 \end{proposition}

\subsection{Some Theorems about Classification
of Toric Codes}

In this paper, we shall classify the toric codes with dimesion equal to $6$, according to the monomial equivalence. Thus, we state the precise definition below.
\begin{definition}\label{df2} Let $C_1$ and $C_2$ be
two codes of block length $n$ and dimension $k$ over $\mathbb{F}_q$.
Let $G_1$ be a generator matrix for $C_1$. Then $C_1$ and $C_2$ are
said to be monomially equivalent if there is an invertible
$n\times n$ diagonal matrix $\Delta$ and an $n\times n$ permutation
matrix $\Pi$ such that
$$
G_2=G_1\Delta\Pi
$$
is a generator matrix for $C_2$.
\end{definition}

It is easy to see that monomial equivalence is actually an
equivalent  relation on codes since a product $\Pi\Delta$ equals
$\Delta^\prime\Pi$ for another invertible diagonal matrix
$\Delta^\prime$. It is also a direct consequence of the definition
that monomially equivalent codes $C_1$ and $C_2$ have the same
dimension and the same minimum distance (indeed, the same full
weight enumerator).

An affine transformation of $\mathbb{R}^m$ is a mapping of
the form $T(x)=Mx+\lambda$, where $\lambda$ is a fixed vector and
$M$ is an $m\times m$ matrix. The affine mappings $T$, where $M\in
GL(m,\mathbb{Z})$ $($so $Det(M)=\pm 1$ $)$  and $\lambda$ have
integer entries, are precisely the bijective affine mappings from
the integer lattice $\mathbb{Z}^m$ to itself.

Generally speaking, it's impractical to determine two given  toric codes to be monomially equivalent directly from the definition. A more practical criteria comes from the nice connection between the monomial equivalence class of the
toric codes $C_P$ and the lattice equivalence class of the polygon $P$ in
\cite{ls2}.
\begin{theorem}\label{th5} If two polytopes $P$
and $\tilde{P}$ are lattice equivalent, then the toric codes $C_{P}$
and $C_{\tilde{P}}$ are monomially equivalent.
\end{theorem}

The definition of the lattice equivalence of two polygons is the following:
\begin{definition}\label{df3} We say that two
integral convex polytopes $P$ and $\tilde{P}$ in $\mathbb{Z}^m$ are lattice equivalent if there exists an invertible integer
affine transformation $T$ as above such that $T(P)=\tilde{P}$.
\end{definition}

For the sake of completeness, we list some simple facts about lattice equivalence of two polytopes $P$
and $\tilde{P}$ in $\mathbb{Z}^2$. 
\begin{proposition}\label{facts of lattice equivalent}
    \indent(i)\quad If $P$ can be transformed to $\tilde{P}$ by
translation, rotation and reflection with respect to x-axis or
y-axis,
then $P$ and $\tilde{P}$ are lattice equivalent;\\
    \indent(ii)\quad If $P$ and $\tilde{P}$ are lattice equivalent, then
they have the same number of sets of $n$
collinear points and the same number of sets of $n$ concurrent segments;\\
    \indent(iii)\quad If $P$ and $\tilde{P}$ are lattice equivalent,
then they are both $n$-side polygons;\\
    \indent(iv)\quad If $P$ and $\tilde{P}$ are lattice equivalent, then
they have the same number of interior integer lattices.
\end{proposition}

These properties are directly followed from Definition \ref{df3}.

Besides the properties of lattice equivalence, Pick's Formula is also a useful tool in the proof of Theorem \ref{th1}. 
\begin{theorem}\textup{(Pick's formula)}
    Assume $P$ is a convex rational polytope in the plane, then
    $$\sharp(P)=A(P)+\frac12\cdot \partial(P)+1,$$
    where $\sharp(P)$ represents the number of lattice points in
    $P$, $A(P)$ is the area of $P$ and $\partial(P)$
    is the perimeter of $P$, with the length of an
    edge between two lattice points defined as one more than the number of
    lattice points lying strictly between them.
\end{theorem}
\begin{remark}
    Generally speaking, $\partial(P)$ is the number of lattice
points on the boundary of $P$. The only exception in plane is line
segment, which should follow the precise definition of length of the
edge above.
\end{remark}

\subsection{Some Theorems to eliminate the upper bound of $q$}

Let us introduce the so-called Hasse-Weil bounds,
which will be used in the proof of Theorem \ref{th2} frequently
to help specifying the exact number of the codewords with some
particular weight, for $q$ large.

\begin{theorem}\label{th7}\textup{(\cite{ap})}
    If $Y$ is an absolutely irreducible but possibly singular curve, $g$
is the \emph{arithmetic genus} of $Y$, $Y(\mathbb{F}_q)$ is the
set of $\mathbb{F}_q$-rational points of curve, then
$$
1+q-2g\sqrt{q}\leq |Y(\mathbb{F}_q)| \leq1+q+2g\sqrt{q}.
$$
These two bounds are called the Hasse-Weil bounds.
\end{theorem}

Let $f\in\mathcal{L}(P)$ and $P_f$ denote its Newton polygon, which is the convex hull of the lattice points in $(\mathbb{F}_q^*)^2$. Denote 
$$f=\sum_{m=(m_1,m_2)\in P_f}\lambda_mx^{m_1}y^{m_2}, \quad \lambda_m\in\mathbb{F}_q^*.$$

Let $X$ be a smooth toric surface over $\overline{\mathbb{F}}_q$ defined by a fan $\Sigma_X\subset\mathbb{R}^2$ which is a refinement of the normal fan of $P_f$. Let $C_f$ be the closure in $X$ of the affine curve given by $f=0$. If $f$ is absolutely irreducible, then $C_f$ is irreducible. By Theorem $\ref{th7}$, 
$$|C_f(\mathbb{F}_q)|\le q+1+2g\sqrt{q},$$
where $g$ is the arithmetic genus of $C_f$.

Let $Z(f)$ be the number of zeros of $f$ in the torus $(\mathbb{F}_q^*)^2$. It is well known that the arithmetic genus $g$ of $C_f$ equals to the number of interior lattice points in $P_f$ (see \cite{ls1} for the curves).

\begin{proposition}\label{pr}
Let $f$ be absolutely irreducible with Newton polygon $P_f$. Then
$$Z(f)\le q+1+2I\left(P_f\right)\sqrt{q},$$
where $I(P_f)$ is the number of interior lattice points.
\end{proposition}
\begin{center}
\section{Proof of the Theorems}
\end{center}

In this section, we shall give the sketch of the proofs of Theorem \ref{th1} and \ref{th2}. Before that, let us clarify the notations first. Let $P_i$ denote an integral convex
polygon in $\mathbb{Z}^2$ with $i$ lattice points, $P_i^{(j)}$ is
the $j^{\textup{th}}$ lattice equivalence class of $P_i$, $V$ is the
additional lattice point, which will be added to $P_i^{(j)}$ and
$P_{i,V}^{(j)}:={\rm conv}\left\{P_i^{(j)},V\right\}$ denote a new integral convex
polygon formed by $P_i^{(j)}$ and $V$. Our strategy is almost the
same as that in \cite{ls2}, by adding all possible choices of $V$ to
$P_5^{(j)}$ to get all lattice equivalence classes of $P_6$ with the help of Pick's formula.

\smallskip

\begin{proof}[Proof of Theorem \ref{th1}.] Let us add $V$ to $P_5^{(1)}={\rm conv}\{(0,0),(4,0)\}$ to see that
$P_6^{(1)}$ and $P_6^{(2)}$ are the only two lattice equivalence
classes. If $V$ is on the x-axis to form a line segment, the only choices of $V$ would be $(5,0)$ or $(-1,0)$, otherwise the new convex polygon have more than $6$ lattice points. Notice that
$P_6^{(1)}={\rm conv}\left\{P_5^{(1)},(5,0)\right\}$ and ${\rm conv}\left\{P_5^{(1)},(-1,0)\right\}$ are lattice
equivalent by translation (i.e. Proposition \ref{facts of lattice
equivalent}, (i)). If $V$ is not on the x-axis, then we have $ \partial\left(P_{5,V}^{(1)}\right)=6$. By
using Pick's formula, $6=\#\left(P_{5,V}^{(1)}\right)=A\left(P_{5,V}^{(1)}\right)
+\frac{1}{2}\partial\left(P_{5,V}^{(1)}\right)+1
=A\left(P_{5,V}^{(1)}\right)+4$,  we get $A\left(P_{5,V}^{(1)}\right)=2$. Therefore, the choices
of $V$ are the lattice points on $y=\pm1$. Say, $V=(x_0,1)$, $x_0$
is integer. By the definition of lattice equivalence, there is an
integer affine transformation
$M=\left(\begin{matrix}1&0\\-x_0&1\end{matrix}\right)$, which
transforms ${\rm conv}\left\{P_5^{(1)},(x_0,1)\right\}$ to $P_6^{(2)}$. The similar
transformation can be found to ${\rm conv}\left\{P_5^{(1)},(x_0,-1)\right\}$.

There are only $14$ lattice equivalence classes $P_6^{(i)}$, $i=1,\cdots,14$, as shown in Fig. \ref{fig-k=6}. Since the arguments are similar, we just list all the possible $V$'s and in which
equivalence class $P_{5,V}^{(i)}$ is, for $i=1,\cdots,7$, in Table \ref{table-1}. The verification is left to the interested readers.

\begin{table}
\caption{}\label{table-1}
\begin{tabular}{|c|c|c|}
\hline add $V$ to $P_5^{(i)}$&possible choice of $V$&
lattice equivalence class ($P_6^{(i)}$)\\
\hline $P_5^{(1)}$&$(5,0)$, $(-1,0)$&$P_6^{(1)}$\\
&$(x_0,\pm1)$, $x_0$ is integer&$P_6^{(2)}$\\
\hline $P_5^{(2)}$&$(4,0)$, $(-1,0)$&$P_6^{(2)}$\\
&$(1,1)$, $(-1,1)$&$P_6^{(3)}$\\
&$(0,-1)$, $(6,-1)$&$P_6^{(4)}$\\
&$(-1,-1)$, $(7,-1)$&$P_6^{(5)}$\\
&$(1,-1)$, $(5,-1)$&$P_6^{(6)}$\\
&$(3,-1)$&$P_6^{(7)}$\\
&$(2,-1)$, $(4,-1)$&$P_6^{(8)}$\\
\hline $P_5^{(3)}$&$(3,0)$, $(-1,0)$&$P_6^{(3)}$\\
&$(-1,-1)$, $(4,-1)$, $(0,-1)$&$P_6^{(9)}$\\
&$(1,2)$, $(-1,2)$, $(4,-1)$&$P_6^{(11)}$\\
&$(1,-1)$, $(2,-1)$&$P_6^{(12)}$\\
&$(0,2)$&$P_6^{(13)}$\\
&$(2,1)$, $(-1,1)$&$P_6^{(14)}$\\
\hline $P_5^{(4)}$&$(3,3)$&$P_6^{(6)}$\\
&$(-1,-1)$&$P_6^{(7)}$\\
&$(-1,0)$, $(0,-1)$&$P_6^{(9)}$\\
&$(-1,1)$, $(1,-1)$&$P_6^{(11)}$\\
&$(1,2)$, $(1,2)$&$P_6^{(12)}$\\
\hline
$P_5^{(5)}$&$(3,0)$&$P_6^{(4)}$\\
&$(-1,0)$&$P_6^{(8)}$\\
&$(1,1)$, $(1,-1)$&$P_6^{(9)}$\\
&$(-1,2)$, $(-1,-2)$&$P_6^{(10)}$\\
&$(-1,1)$, $(-1,-1)$&$P_6^{(11)}$\\
\hline $P_5^{(6)}$&$(0,3)$&$P_6^{(5)}$\\
&$(0,-1)$&$P_6^{(6)}$\\
&$(1,1)$, $(-1,0)$&$P_6^{(11)}$\\
\hline $P_5^{(7)}$&$(2,0)$, $(-2,0)$, $(0,2)$,
$(0,-2)$&$P_6^{(8)}$\\
&$(1,1)$, $(1,-1)$, $(-1,-1)$, $(-1,1)$&$P_6^{(12)}$\\
\hline
\end{tabular}
\end{table}
\end{proof}

In order to show Theorem \ref{th2}, we only need to determine whether two toric
surface codes constructed from the polygons in Fig. \ref{fig-k=6}
can be pairwise monomially equivalent, due to the result in \cite{ls2} and \cite{yz} that no two of the toric codes $C_P(\mathbb{F}_q)$, $q>5$ constructed from the polytopes with dimension $k=3,4$ and $5$ are monomially equivalent. Our strategy is the following: 
\begin{enumerate}
	\item For $q$
small, say $q\leq8$, we use the GAP code (with toric package and guava package) to get their enumerator polynomials directly. If those
enumerator polynomials of $C_{P_6^{(i)}}$, $1\leq i\leq14$, are
different from each other on $\mathbb{F}_q$, for $7\leq q\leq8$,
then they are not pairwirse monomially equivalent. If they are the same
in some cases (for example, $C_{P_6^{(5)}}$ and $C_{P_6^{(6)}}$ over
$\mathbb{F}_7$, $C_{P_6^{(4)}}$ and $C_{P_6^{(5)}}$ over
$\mathbb{F}_8$, see Table \ref{table-a.1} in the appendix), we need some further
investigations. 
	\item For $q$ large, say $q\ge9$, we shall compare the
invariants of the codes, including minimum distance, the number of
the codewords with some particular weight, etc. Once we
could identify one invariant in one case to be different from that in another case, then
we conclude that they are pairwise monomially inequivalent. However, the
estimate of the number of the codewords with some particular weight
depends on how large $q$ is. Therefore, we still need to use GAP (with toric package and guava package) for small $q$ (see Table \ref{table-a.2} and \ref{table-a.3} in the appendix).
\end{enumerate}

The first step in our strategy is to tell the monomial equivalence
of $C_{P_6^{(i)}}$, $1\leq i\leq14$, for $q\leq8$. It's easy to see
from Table \ref{table-a.1} in the appendix that all the enumerator polynomials of
$C_{P_6^{(i)}}$, $1\leq i\leq14$, are different, except that of
$C_{P_6^{(5)}}$ and $C_{P_6^{(6)}}$ over $\mathbb{F}_7$ and that of
$C_{P_6^{(4)}}$ and $C_{P_6^{(5)}}$ over $\mathbb{F}_8$. Here, an
interesting phenomena occurs. Two toric codes constructed from two
lattice non-equivalent polygons could also be monomially equivalent.
$C_{P_6^{(5)}}$ and $C_{P_6^{(6)}}$ over $\mathbb{F}_7$ is the typical example.
\begin{proposition}\label{over F7}
    $C_{P_6^{(5)}}$ and $C_{P_6^{(6)}}$ over $\mathbb{F}_7$ are monomially
    equivalent.
\end{proposition}
\begin{proof}
We use the Magma program to give the generator matrices of these two toric
codes over $\mathbb{F}_7$. For the Magma code, please refer to \cite{j}.
\end{proof}

Unfortunately, the other pair $C_{P_6^{(4)}}$ and $C_{P_6^{(5)}}$ over
$\mathbb{F}_8$ can't be determined by the same way in Proposition
\ref{over F7}, since the command ``IsEquivalent" in Magma can only
be used to compare toric codes over $\mathbb{F}_q$ with $q=4$ or
small prime numbers. Moreover, it is infeasible to show the monomial equivalence directly from the definition. So we leave the problem open. Based on the result in Proposition \ref{over F7} and the fact that the enumerator polynomial of $C_{P_6^{(4)}}$ and
$C_{P_6^{(5)}}$ over $\mathbb{F}_8$ is exactly the same, we 
conjecture that this pair, i.e. $C_{P_6^{(4)}}$ and $C_{P_6^{(5)}}$ over
$\mathbb{F}_8$, is also monomially equivalent. Further, we ask a more general question for the interested readers: For which $q$ and $k$ are there monomially equivalent toric codes over $\mathbb{F}_q$ from polytopes that are not lattice equivalent?

Next, we shall classify $C_{P_6^{(i)}}$, $1\leq i\leq14$, for
$q\ge9$. The first invariant to be examined  is the minimum distance
(or the minimum weight), denoted as $d\left(C_{P_6^{(i)}}\right)$.
\begin{proposition}\label{d}
    According to $d\left(C_{P_6^{(i)}}\right)$, $1\leq
    i\leq14$, for $q\ge9$, no code in one of the five groups is monomially equivalent to a code in any of the other four groups:
	\begin{enumerate}
    	\item[(\romannumeral1)] $C_{P_6^{(1)}}$;
    	\item[(\romannumeral2)] $C_{P_6^{(2)}}$;
    	\item[(\romannumeral3)] $C_{P_6^{(14)}}$;
    	\item[(\romannumeral4)] $C_{P_6^{(i)}}$, for $3\leq i\leq8$;
    	\item[(\romannumeral5)] $C_{P_6^{(i)}}$, for $9\leq i\leq13$.
	\end{enumerate}
\end{proposition}
\begin{proof}
	 It follows directly from Theorem \ref{th4} and Theorem \ref{sjx} that for all $q$, we have $d\left(C_{P_6^{(1)}}\right)=(q-1)^2-5(q-1)$,
$d\left(C_{P_6^{(14)}}\right)=(q-1)^2-(3q-5)$,
$d\left(C_{P_6^{(2)}}\right)=(q-1)^2-4(q-1)$ and
$d\left(C_{P_6^{(13)}}\right)=(q-1)^2-2(q-1)$. Besides these four, we still
need to compute $d\left(C_{P_6^{(i)}}\right)$, $3\leq i\leq12$.

For $C_{P_6^{(3)}}$, it is a subcode of $C_{P_{3,3}^\triangle}$ with
$d\left(C_{P_{3,3}^\triangle}\right)=(q-1)^2-3(q-1)$, by Theorem \ref{sjx};
while it is also a supercode of $C_{P_{3,1}^\triangle}$ with the same
minimum distance as $C_{P_{3,3}^\triangle}$, by Theorem \ref{th4}.
Therefore, $d\left(C_{P_6^{(3)}}\right)=(q-1)^2-3(q-1)$ for all $q$.

For $d\left(C_{P_6^{(i)}}\right)$, $4\leq i\leq12$, they can be figured out in
the similar way for each $i$ by Proposition \ref{bound}. We illustrate
the argument for $d\left(C_{P_6^{(4)}}\right)$ in detail and leave the similar
work for $d\left(C_{P_6^{(i)}}\right)$, $5\leq i\leq12$, to the interested
readers. 
For $C_{P_6^{(4)}}$, $L\left(P_6^{(4)}\right)=3$, there is no factorization $f=f_1f_2f_3$ containing an exceptional triangle, $c=A/2-L+11/4=5/4$. For $q\ge\max\left(37,(c+\sqrt{c^2+2})^2\right)=37$, we get $d\left(C_{P_6^{(4)}}\right)\geq(q-1)^2-3(q-1)$. It is
easy to see that $C_{P_6^{(4)}}$ indeed contains the codewords,
which has $3(q-1)$ zeros in $(\mathbb{F}_q^*)^2$. For example, the codewords
come from the evaluation ${\rm ev}(d(x-a)(x-b)(x-c))$, where $a,b,c,d\in
\mathbb{F}_q^*$ and $a\neq b\neq c$. Therefore,
$d\left(C_{P_6^{(4)}}\right)=(q-1)^2-3(q-1)$. With the similar argument, we
conclude that when $q\ge37$, $d\left(C_{P_6^{(i)}}\right)=(q-1)^2-3(q-1)$, for $5\leq
i\leq 8$, and $d\left(C_{P_6^{(i)}}\right)=(q-1)^2-2(q-1)$, for $9\leq
i\leq12$. Using GAP (with toric package and guava package), the minimum distance for $7\le q<37$ can be computed: $d\left(C_{P_6^{(i)}}\right)=(q-1)^2-3(q-1)$, for $5\leq i \leq 8$ and $d\left(C_{P_6^{(i)}}\right)=(q-1)^2-2(q-1)$, for $9\leq
i\leq12$ for all $q\ge9$.

Since the minimum distance is one of the invariants of monomial equivalence of
codes, we summerize for $q\ge9$:
\begin{enumerate}
\item[(\romannumeral1)]$d\left(C_{P_6^{(1)}}\right)=(q-1)^2-5(q-1)$,
\item[(\romannumeral2)]$d\left(C_{P_6^{(2)}}\right)=(q-1)^2-4(q-1)$,
\item[(\romannumeral3)]$d\left(C_{P_6^{(14)}}\right)=(q-1)^2-(3q-5)$,
\item[(\romannumeral4)]$d\left(C_{P_6^{(i)}}\right)=(q-1)^2-3(q-1)$ for $3\leq i\leq 8$, 
\item[(\romannumeral5)]$d\left(C_{P_6^{(i)}}\right)=(q-1)^2-2(q-1)$ for $9\leq i\leq 13$.
\end{enumerate}
\end{proof}

Just according to the minimum distance, the monomial equivalence/inequivalence of any
two codes both from either group $\rm (\romannumeral4)$ or $\rm (\romannumeral5)$ in Proposition \ref{d} are still unknown. We shall examine two more invariants: the numbers of
the codewords with weight $(q-1)^2-2(q-1)$ and $(q-1)^2-(2q-3)$,
denote as $n_1\left(C_{P_6^{(i)}}\right)$ and $n_2\left(C_{P_6^{(i)}}\right)$, respectively.
We start with group $(\romannumeral4)$, i.e.  $C_{P_6^{(i)}}$, $3\leq i\leq8$. The group $(\romannumeral5)$, i.e. $C_{P_6^{(i)}}$, $9\leq i\leq13$, can be done similarly.

\begin{table}[!tbh]
\caption{}\label{table-2}
\begin{tabular}{|l|l|l|l|}
    \hline
    $C_{P_6^{(i)}}$& Distinct families of reducible polynomials&
        Number of codewords&$n_1\left(C_{P_6^{(i)}}\right)$\\
\hline
    $C_{P_6^{(3)}}$&$c(x-a)(x-b)$, $a, b, c\in \mathbb{F}_q^{*}$, $a\neq
    b$&${{q-1}\choose{2}}(q-1)$&$=4{{q-1}\choose{2}}(q-1)$\\
    &$cx(x-a)(x-b)$, $a, b,
c\in \mathbb{F}_q^{*}$, $a\neq b$&${{q-1}\choose{2}}(q-1)$&\\
    &$c(x-a)^2(x-b)$, $a, b, c\in \mathbb{F}_q^{*}$,
$a\neq b$&$2{{q-1}\choose{2}}(q-1)$&\\
\hline

    $C_{P_6^{(5)}}$&$c(x-a)(x-b), a, b, c\in \mathbb{F}_q^{*}, a\neq
    b$&${{q-1}\choose{2}}(q-1)$&$=4{{q-1}\choose{2}}(q-1)$\\
    &$c x(x-a)(x-b), a, b, c\in \mathbb{F}_q^{*}, a\neq
    b$&${{q-1}\choose{2}}(q-1)$&\\
    &$c(x-a)^2(x-b), a, b, c\in
\mathbb{F}_q^{*}, a\neq b$&$2{{q-1}\choose{2}}(q-1)$&\\
\hline

    $C_{P_6^{(6)}}$&$c(xy-a)(xy-b), a, b, c\in \mathbb{F}_q^{*}, a\neq
    b$&${{q-1}\choose{2}}(q-1)$&$=4{{q-1}\choose{2}}(q-1)$\\
    &$c xy(xy-a)(xy-b), a, b, c\in \mathbb{F}_q^{*}, a\neq b$&${{q-1}\choose{2}}(q-1)$&\\
    &$c(xy-a)^2(xy-b), a, b, c\in
\mathbb{F}_q^{*}, a\neq b$&$2{{q-1}\choose{2}}(q-1)$&\\
\hline

    $C_{P_6^{(7)}}$&$c(x-a)(x-b), a, b, c\in \mathbb{F}_q^{*}, a\neq
    b$&${{q-1}\choose{2}}(q-1)$&if $3|(q-1)$,\\
    &$c x(x-a)(x-b), a, b, c\in \mathbb{F}_q^{*},a\neq
    b$&${{q-1}\choose{2}}(q-1)$&$=4{{q-1}\choose{2}}(q-1)+$\\
    &$c(x-a)^2(x-b), a, b, c\in \mathbb{F}_q^{*}, a\neq
    b$&$2{{q-1}\choose{2}}(q-1)$&$\frac{2}{3}(q-1)^3$;\\
    &$\textup{if}\,\, 3|(q-1), 
cy^{-1}(y-a)(x^3-by), a,b, c\in \mathbb{F}_q^{*},$&&if $3\nmid q-1$,\\
    &$ \phantom{if 3|(q-1)}\textup{and}\,
ab\neq i^3$ for all $i\in \mathbb{F}_q^*$ &$\frac{2}{3}(q-1)^3$&$=4{{q-1}\choose{2}}(q-1)$\\
\hline

    $C_{P_6^{(4)}}$& $c(x-a)(x-b), a, b, c\in \mathbb{F}_q^{*}, a\neq
    b$&${{q-1}\choose{2}}(q-1)$&$=5{{q-1}\choose{2}}(q-1)$\\
    &$cx(x-a)(x-b), a, b, c\in \mathbb{F}_q^{*}, a\neq b$&${{q-1}\choose{2}}(q-1)$&\\
    &$c(x-a)^2(x-b), a, b, c\in \mathbb{F}_q^{*}, a\neq b$&$2{{q-1}\choose{2}}(q-1)$&\\
    &$cy^{-1}(y-a)(y-b), a, b, c\in \mathbb{F}_q^{*}, a\neq b$&${{q-1}\choose{2}}(q-1)$&\\
\hline

    $C_{P_6^{(8)}}$&$c(x-a)(x-b), a, b, c\in \mathbb{F}_q^{*}, a\neq
    b$&${{q-1}\choose{2}}(q-1)$&if $q\neq 2^n$, $n\in\mathbb{Z}_+$, \\
    &$c x^{-1}(x-a)(x-b), a, b, c\in \mathbb{F}_q^{*}, a\neq
    b$&${{q-1}\choose{2}}(q-1)$&$=5{{q-1}\choose{2}}(q-1)+$ \\
    &$cx^{-1}(x-a)^2(x-b), a, b, c\in \mathbb{F}_q^{*}, a\neq
    b$&$2{{q-1}\choose{2}}(q-1)$&$\frac{1}{2}(q-1)^3;$\\
    &$c y^{-1}(y-a)(y-b), a, b, c\in \mathbb{F}_q^{*},a\neq
    b$&${{q-1}\choose{2}}(q-1)$&if $q=2^{n}$, $n\in\mathbb{Z}_+$,\\
    &if $q\neq
2^{n}, cx^{-1}y^{-1}(y-ax)(b-xy),$&& $=5{{q-1}\choose{2}}(q-1)$\\
    & \phantom{if $q\neq
2^{m}$,}$n\in\mathbb{Z}_+$,$a, b, c\in
\mathbb{F}_q^{*}$,&&\\
    & \phantom{if }$a\neq b, \frac{a}{b}\neq\alpha^{2i}$ for
$1\leq i\leq q-2$&$\frac{1}{2}(q-1)^3$&\\
\hline
\end{tabular}
\end{table}

The basic idea to examine the pairwise monomial inequivalence of any
two codes $C_{P_6^{(i)}}$, $3\leq i\leq8$, is: 
\begin{enumerate}
	\item to find out
$n_1\left(C_{P_6^{(i)}}\right)$, $3\leq i\leq8$ and to sort the codes with the same $n_1\left(C_{P_6^{(i)}}\right)$ into subgroups to be determined later;
	\item to give the range of $n_2\left(C_{P_6^{(i)}}\right)$
among the codes with the same $n_1\left(C_{P_6^{(i)}}\right)$ and to compare them
to give the final classification. 
\end{enumerate}
Fortunately, in our situation,
these two invariants are enough to give a complete classification of monomial
equivalence class to $C_{P_6^{(i)}}$, $3\leq i\leq8$.

\begin{table}[!h]
\caption{}\label{table-3}
\begin{tabular}{|l|l|l|l|}
    \hline
    $C_{P_6^{(i)}}$&Distinct families of reducible polynomials&
        $\#$ of codewords&$n_2\left(C_{P_6^{(i)}}\right)$\\
\hline

    $C_{P_6^{(4)}}$, &None&$0$&$=0$\\
    $q=2^n$, $n\in\mathbb{Z}_+$&&&\\
\hline

    $C_{P_6^{(8)}}$,&$cx^{-1}y^{-1}(y-ax)(b-xy)$ and $a^{-1}b=\alpha^2$ &$(q-1)^3$&$(q-1)^3$\\
    \quad$q=2^n$,&for one unique $\alpha\in\mathbb{F}_q^*$&&\\
    \quad$m\in\mathbb{Z}_+$&&&\\
\hline
    $C_{P_6^{(3)}}$&$d(x-a)(y-bx-c),a,d\in\mathbb{F}_q^*$,&&\\
    &\phantom{d(x-a)}if $b=0$, $c\neq0$&$(q-1)^3$&$=6(q-1)^3$\\
    &\phantom{d(x-a)}or $c=0$, $b\neq0$&$(q-1)^3$&$+{{q}\choose{2}}(q-1)^3$;\\
    &\phantom{d(x-a)}or $b, c \neq 0$ and
    $a=-\frac{c}{b}$&$(q-1)^3$&\\
    &$a(y+bx^2+cx+d)(x+e),a,b,c,d,e\in\mathbb{F}_q^*$,&$\frac{1}{2}(q-1)^5$, $q\neq 2^n$&\\
    &~if $bx^2+cx+d$ is absolutely irreducible &${{q}\choose{2}}(q-1)^3$, $q=2^n$&\\
    &~or $bx^2+cx+d=b(x+e)^2$ when $q\neq 2^n$;&$(q-1)^3$, $q\neq 2^n$&\\
    &$a(y+bx^2+cx)(x+b^{-1}c), a,b,c\in\mathbb{F}_q^*$&$(q-1)^3$&\\
    
    &$a(y+bx^2-c)(x+d)$&&\\
    &~if $q\neq 2^n$ for all $n\in\mathbb{Z_+}$ &&\\
    &~and $b^{-1}c\neq \alpha^2$ for all $\alpha\in\mathbb{F}_q^*$&$\frac{1}{2}(q-1)^4$, $q\neq 2^n$&\\
    &~or if $q=2^n$ for some $n\in\mathbb{Z}_+$&$(q-1)^3$, $q=2^n$&\\
    &~ and $(-d)^2=b^{-1}c$&&\\
    &$a(y+bx^2)(x+c)$&$(q-1)^3$&\\
    
\hline

    $C_{P_6^{(5)}}$,&&&\\
    $q\geq27$ &None&$0$&$=0\ {^\dag}$\\
\hline

    $C_{P_6^{(6)}}$&$c(x-a)(y-b),a,b,c\in\mathbb{F}_q^*$&$(q-1)^3$&$=(q-1)^3$\\
\hline
    $C_{P_6^{(7)}}$,&c$y^{-1}(y-a)(y-bx^3)$&$(q-1)^3$&\\
    \quad $3\nmid(q-1)$&$cy^{-1}(y-ax^2)(y-bx),a, b, c\in
    \mathbb{F}_q^*$&$(q-1)^3$&$=2(q-1)^3$\\
\hline
\end{tabular}
	\begin{flushleft}
	{\footnotesize $^\dag$ Only if $q\geq47$, $n_2\left(C_{P_6^{(5)}}\right)$ could
be shown to be $0$. As a supplement, Table \ref{table-a.3} in the appendix illustrates
that $n_2\left(C_{P_6^{(5)}}\right)=0$, for $27\leq q\leq43$. When $q=25$,
although $n_2\left(C_{P_6^{(5)}}\right)$ is not zero, the exact enumerator
polynomials of $P_6^{(3)}$, $P_6^{(5)}$ and $P_6^{(6)}$ are listed
in Table \ref{table-a.3} explicitly.}
	\end{flushleft}
\end{table}

To be more precise, the way to compute $n_1\left(C_{P_6^{(i)}}\right)$ is
to enumerate the families of evaluations that contribute to
weight $(q-1)^2-2(q-1)$. The completeness of the enumeration above
is followed by Theorem \ref{th7}, which requires 
$q$ being large, say $q\geq23$ (this lower bound is given by the inequality in Proposition \ref{pr}) in most of the cases. Then, we
use the GAP code (with toric package and guava package) again to make up the gap $9\leq
q\leq19$, see Table \ref{table-a.2}. For $q\geq23$, with the help of
$n_1\left(C_{P_6^{(i)}}\right)$, we can exclude some codes and sort
the ones left into several subgroups with the same
$n_1\left(C_{P_6^{(i)}}\right)$. Then, by enumerating the families of
evaluations that contribute to weight $(q-1)^2-(2q-3)$, the range of
$n_2\left(C_{P_6^{(i)}}\right)$ can be obtained to classify the subgroups.
\begin{proposition}\label{n_1}
    For $q\ge9$, no two codes of $C_{P_6^{(i)}}$, $3\leq i\leq8$, are pairwise monomially equivalent.
\end{proposition}
\begin{proof}
	 We shall identify $n_1\left(C_{P_6^{(i)}}\right)$, for
$3\leq i\leq8$, one by one. Since the arguments are similar to that in the proof of Theorem 6 in \cite{ls2}, we shall
just investigate $n_1\left(C_{P_6^{(3)}}\right)$ in detail, for the sake of completeness; while for $n_1\left(C_{P_6^{(i)}}\right)$, $4\leq
i\leq8$, only the key differences will be stated in Table \ref{table-2}.
Interested readers can complete the arguments for
$n_1\left(C_{P_6^{(i)}}\right)$, $4\leq i\leq8$.

For $C_{P_6^{(3)}}$,
$n_1\left(C_{P_6^{(3)}}\right)\geq4{{q-1}\choose{2}}(q-1)$, because there are
three distinct families of reducible polynomials:
\begin{itemize}
     \item $c(x-a)(x-b)$, $a, b, c\in \mathbb{F}_q^{*}$, $a\neq b$, has
${{q-1}\choose{2}}(q-1)$ such codewords,
     \item $cx(x-a)(x-b)$, $a, b,
c\in \mathbb{F}_q^{*}$, $a\neq b$, has ${{q-1}\choose{2}}(q-1)$
codewords,
 	\item $c(x-a)^2(x-b)$, $a, b, c\in \mathbb{F}_q^{*}$,
$a\neq b$, has $2{{q-1}\choose{2}}(q-1)$ codewords.
\end{itemize}
Actually, we claim that there are exactly
4${q-1}\choose{2}$$(q-1)$ such codewords in $C_{P_6^{(3)}}$. Any
other such codewords could only come from evaluating a linear
combination of $\{1, x, x^2, x^3, y, xy\}$ in which at least one of $y$ and
$xy$ appears with nonzero coefficients (otherwise, we are in
a case covered previously).
\begin{table}[!h]
\caption{}\label{table-4}
\begin{tabular}{|l|l|l|l|}
\hline
    $C_{P_6^{(i)}}$& Distinct families of reducible polynomials&
        Number of codewords&$n_1\left(C_{P_6^{(i)}}\right)$\\
\hline
    $C_{P_6^{(9)}}$& $c(x-a)(x-b)$, $a, b, c\in \mathbb{F}_q^{*}$,
$a\neq b$&${{q-1}\choose{2}}(q-1)$&$=2{{q-1}\choose{2}}(q-1)$\\
    &$cy^{-1}(y-a)(y-b)$, $a, b,
c\in \mathbb{F}_q^{*}$, $a\neq b$&${{q-1}\choose{2}}(q-1)$&\\
\hline
    $C_{P_6^{(11)}}$&$c(y-a)(y-b), a, b, c\in \mathbb{F}_q^{*}, a\neq b,$&${{q-1}\choose{2}}(q-1)$&$=2{{q-1}\choose{2}}(q-1)$\\
    &$c x^{-1}y^{-1}(xy-a)(xy-b), a, b, c\in
\mathbb{F}_q^{*}, a\neq b$&${{q-1}\choose{2}}(q-1)$&\\
\hline
    $C_{P_6^{(12)}}$&$cy^{-1}(y-a)(y-b), a, b, c\in \mathbb{F}_q^{*},a\neq b$&${{q-1}\choose{2}}(q-1)$&if $q=2^m$,\\
    &$c x^{-1}(x-a)(x-b), a, b, c\in
\mathbb{F}_q^{*}, a\neq b$&${{q-1}\choose{2}}(q-1)$&\quad$=2{{q-1}\choose{2}}(q-1)$;\\
    &If $q\neq2^m,
cx^{-1}y^{-1}(y-ax)(b-xy),$&&if $q\neq2^m$,\\
    &\phantom{If $q\neq2^m$}$m\in\mathbb{Z}_+$,$a, b, c\in \mathbb{F}_q^{*},$&&\quad$=2{{q-1}\choose{2}}(q-1)+$\\
    &\phantom{If}$a\neq b,
\frac{a}{b}\neq\alpha^{2i}$, for $1\leq i\leq q-2$&$\frac{1}{2}(q-1)^3 $&\quad$\frac{1}{2}(q-1)^3 $ \\
\hline
    $C_{P_6^{(10)}}$&$c(x-a)(x-b), a, b, c\in \mathbb{F}_q^{*}, a\neq b$&${{q-1}\choose{2}}(q-1)$&$=3{{q-1}\choose{2}}(q-1)\ {^\dag}$\\
    &$c y^{-1}(y-a)(y-b), a, b, c\in
\mathbb{F}_q^{*},a\neq b$&${{q-1}\choose{2}}(q-1)$&\\
    &$c x^{-1}(y-ax)(y-bx), a, b, c\in
\mathbb{F}_q^{*}, a\neq b$&${{q-1}\choose{2}}(q-1)$&\\
\hline
    $C_{P_6^{(13)}}$&$c(y-a)(y-b), a, b, c\in \mathbb{F}_q^{*},
a\neq
    b$&${{q-1}\choose{2}}(q-1)$&$=3{{q-1}\choose{2}}(q-1)$\\
    &$c(x-a)(x-b), a, b, c\in \mathbb{F}_q^{*}, a\neq
    b$&${{q-1}\choose{2}}(q-1)$&\\
    &$c(x-ay)(x-by), a, b, c\in
\mathbb{F}_q^{*}, a\neq b$&${{q-1}\choose{2}}(q-1)$&\\
\hline
\end{tabular}
\begin{flushleft}
	{\footnotesize $^\dag$ only if $q\geq43$,
$n_1\left(C_{P_6^{(10)}}\right)=3{{q-1}\choose{2}}(q-1)$ could be proven. As a
supplement, Table \ref{table-a.3} in the appendix illustrates that
$n_1\left(C_{P_6^{(10)}}\right)$ is still $3{{q-1}\choose{2}}(q-1)$, when
$23\leq q\leq41$.}
\end{flushleft}
\end{table}

If either the coefficient of $y$ or that of $xy$ is zero, such polynomial will be absolutely irreducible. Indeed, if $y$ has nonzero coefficient and $xy$ doesn't, then by Proposition $\ref{pr}$, the number of zeros  of such polynomial $f$ in the torus $(\mathbb{F}_q^*)^2$ has a bound:$$Z(f)\le q+1+2I(P_f)\sqrt{q}\le q+1+2I\left(P_6^{(3)}\right)\sqrt{q}=q+1+0=q+1.$$
$q+1<2q-2$ for all $q>3$, so such polynomial can never have $2q-2$ zeros. If $xy$ has nonzero coefficient and $y$ doesn't, similarly argument implies that such polynomial can never have $2q-2$ zeros either. 

If both of them have nonzero coefficients, the only possible cases for the polynomial to be reducible are that : 
\begin{enumerate}
	\item[(\uppercase\expandafter{\romannumeral1})] $a(y+bx^2+cx+d)(x+e)$, $a,b,c,d,e\in\mathbb{F}_q^{*}$. It has zeros of two types: (1) $(-e, j)$ for $j\in \mathbb{F}_q^*$; (2)$(i,j)$ such that $j+bi^2+ci+d=0$. In the first type, there are $q-1$ zeros. In the second type, for every $i\in \mathbb{F}_q^*$, there is at most one $j$ such that $j+bi^2+ci+d=0$.

 Actually, if $bx^2+cx+d$ is absolutely irreducible, then there are exactly $q-1$ zeros of second type. In this case, there is exactly one common zero $(-e,j)$ of both types, where $j+be^2-ce+d=0$. Thus the polynomial has exactly $2q-3$ zeros.
 
 If $bx^2+cx+d=b(x+e_1)(x+e_2)$ where $e_1\neq e_2$ and $e_1,e_2\neq e$, then there are exactly $q-3$ zeros of second type. And there is exactly one common zero $(-e,j)$ of both types. Thus the polynomial has exactly $2q-5$ zeros.
 
 If $bx^2+cx+d=b(x+e_1)^2$ where $e_1\neq e$, then there are exactly $q-2$ zeros of second type. And there is exactly one common zero $(-e,j)$ of both types. Thus the polynomial has exactly $2q-4$ zeros.
 
 If $bx^2+cx+d=b(x+e)(x+e_1)$ where $e_1\neq e$, then there are exactly $q-3$ zeros of second type. And there is no common zeros of both types. Thus the polynomial has exactly $2q-4$ zeros.
 
 If $bx^2+cx+d=b(x+e)^2$,  then there are exactly $q-2$ zeros of second type. And there is no common zeros of both types. Thus the polynomial has exactly $2q-3$ zeros.

 So this kind of polynomials cannot have $2(q-1)$ zero points. 
 
\item[(\uppercase\expandafter{\romannumeral2})] $a(y+bx^2+cx)(x+d), a,b,c,d\in\mathbb{F}_q^*$. It has zeros of two types: (1)$(-d,j)$ for $j\in\mathbb{F}_q^*$; (2)(i,j) such that $j+bi^2+ci=0$. In the first type, there must be $q-1$ zeros. In the second type, there must be $q-2$ zeros. 
 
 When $d=b^{-1}c$, there is no common zero of both types. Thus the polynomial has exactly $2q-3$ zeros.
 
 When $d\neq b^{-1}c$, then there is exactly one common zero $(-d,j_0)$ of both types, where $j_0+bd^2-cd=0$. Thus the polynomial  has exactly $2q-4$ zeros.
 
\item[(\uppercase\expandafter{\romannumeral3})] $a(y+bx^2-c)(x+d), a,b,c,d\in\mathbb{F}_q^*$. It has zeros of two types: (1) $(-d,j)$ for $j\in\mathbb{F}_q^*$; (2) (i,j) such that $j+bi^2-c=0$. In the first type, there must be $q-1$ zeros. 
 
 If $q\neq 2^n$ for any $n\in\mathbb{Z}_+$, denote $\mathbb{F}_q^*$ as $\{\alpha^1,\alpha^2, \cdots, \alpha^q-1=1\}$. When $b^{-1}c=\alpha^{2k}$ for some $k\in\{1,2,\cdots,\frac{q-1}{2}\}$, then there are $q-3$ zeros of second type.  Thus the polynomial has at most $2q-4$ zeros. When $b^{-1}c\neq i^2$ for any $i\in\mathbb{F}_q^*$, then there are $q-1$ zeros of second type and there is exactly one common zero $(-d,j_0)$ of both types, where $j_0+bd^2-c=0$. Thus the polynomial has exactly $2q-3$ zeros.
 
 If $q=2^n$ for some $n\in\mathbb{Z}_+$, then there is one unique $i_0$ such that $i_0^2=b^{-1}c$. So there are $q-2$ zeros of second type. If $-d=i_0$, then there is no common zero of both types. Thus the polynomial has exactly $2q-3$ zeros. If $-d\neq i_0$, then the polynomial has exactly $2q-4$ zeros. 
 
\item[(\uppercase\expandafter{\romannumeral4})] $a(y+bx^2)(x+c), a,b,c\in\mathbb{F}_q^*$. It has exactly $2q-3$ zeros.
 
\item[(\uppercase\expandafter{\romannumeral5})] $a(x+by+c)(x+d)$. It has zeros of two types: (1)$(-d,j)$ for $j\in\mathbb{F}_q^*$; (2)$(i,j)$ such that $i+bj+c=0$ where $i,j\in\mathbb{F}_q^*$. In the first case, there are $q-1$ zeros. In the second type, for every $j\in\mathbb{F}_q^*$, there are at most one $i$ such that $i=-bj+c$. 

Let $j=-b^{-1}c$, then there is no $i$ such that $(i,-b^{-1}c)$ is a zero in the second type. So there are $q-2$ zeros in the second type of zeros. 

If $d\neq c$, then$(-d, b^{-1}(d-c))$ is a common zero of both types. So such polynomial has at exactly $2q-4$ zeros. 

  if $d=c$, then there is no common zeros of both types. So such polynomial has exactly $2q-3$ zeros.  

\item[(\uppercase\expandafter{\romannumeral6})] $a(x+by)(x+c), a,b,c\in\mathbb{F}_q^*$. It has $2q-3$ zeros.

\item[(\uppercase\expandafter{\romannumeral7})] $c(x+a)(y+b)$. It has $2q-3$ zeros.

\end{enumerate}
In a word, such reducible polynomials cannot have $2q-2$ zeros.

\begin{table}[!h]
\caption{}\label{table-5}	
\begin{tabular}{|l|l|l|l|}
\hline
    $C_{P_6^{(i)}}$&Distinct families of reducible polynomials&
        Number of codewords&$n_2\left(C_{P_6^{(i)}}\right)$\\
\hline
    $C_{P_6^{(10)}}$,&&&\\
    \quad$q\geq27$&None&$0$&$=0\ {^\dag}$\\
\hline
    $C_{P_6^{(13)}}$&$c(x-a)(y-b),a, b, c\in \mathbb{F}_q^{*}$&$>0$&$>0$\\
\hline
    $C_{P_6^{(9)}}$&$d(x-a)(y-bx-c)$,&&\\
    &\phantom{d(x-a)}if $b=0, a, c, d\in\mathbb{F}_q^{*}$&$(q-1)^3$&$=5(q-1)^3$\\
    &\phantom{d(x-a)}or $c=0, a, b, d\in \mathbb{F}_q^{*}$&$(q-1)^3$&\\
    &\phantom{d(x-a)}or $a, b, c, d\in\mathbb{F}_q^{*}, a=-\frac{c}{b}$&$(q-1)^3$&\\
    &$c y^{-1}(y-a)(xy-b), a, b, c\in\mathbb{F}_q^{*}$&$(q-1)^3$&\\
    &$c y^{-1}(y-a)(xy-dy-b)$,&&\\
    &\phantom{d(x-a)}$a, b, c, d\in\mathbb{F}_q^{*},d=-\frac{b}{a}$&$(q-1)^3$&\\
\hline
     $C_{P_6^{(11)}}$&$d(y-a)(x-by-c)$,&&\\
    &\phantom{d(x-a)}if $b=0, a, c,
    d\in\mathbb{F}_q^{*}$&$(q-1)^3$&$=3(q-1)^3$\\
    &\phantom{d(x-a)}or $c=0, a, b, d\in \mathbb{F}_q^{*}$&$(q-1)^3$&\\
    &\phantom{d(x-a)}or $a, b, c, d\in\mathbb{F}_q^{*},
    a=-\frac{c}{b}$&$(q-1)^3$&\\
\hline
    $C_{P_6^{(12)}}$,&$cx^{-1}(x-a)(y-b), a, b,
    c\in\mathbb{F}_q^{*}$&$(q-1)^3$&$\geq6(q-1)^3$\\
    \quad$q=2^m$,&$c x^{-1}(xy-a)(x-b), a, b,
    c\in\mathbb{F}_q^{*}$&$(q-1)^3$&\\
    \quad$m\in\mathbb{Z}_+$&$c y^{-1}(xy-a)(y-b), a, b, c\in\mathbb{F}_q^{*}$&$(q-1)^3$&\\
    &$c y^{-1}(y-a)(xy-dy-b), a, b, c,
    d\in\mathbb{F}_q^{*}$,&&\\
    &\phantom{d(x-a)}$d=-\frac{b}{a}$&$(q-1)^3$&\\
    &$c x^{-1}(x-a)(xy-dx-b), a, b, c, d\in\mathbb{F}_q^{*}$,&&\\
    &\phantom{d(x-a)}$d=-\frac{b}{a}$&$(q-1)^3$&\\
    &$c x^{-1}y^{-1}(xy-a)(y-bx), a, b, c,
    d\in\mathbb{F}_q^{*}$&$(q-1)^3$&\\
\hline
\end{tabular}
\begin{flushleft}
	{\footnotesize $^\dag$ Only if $q\geq47$, $n_2(C_{P_6^{(10)}})$
could be shown to be $0$. As a supplement, Table \ref{table-a.3} in the appendix
illustrates that $n_2\left(C_{P_6^{(10)}}\right)=0$, for $27\leq q\leq41$. When
$q=25$, although $n_2\left(C_{P_6^{(10)}}\right)$ is not zero, the exact
enumerator polynomials of $P_6^{(10)}$ and $P_6^{(13)}$ are listed
in Table \ref{table-a.3} explicitly.}
\end{flushleft}
\end{table}

When the polynomial is absolutely irreducible, by  Proposition $\ref{pr}$, we obtain that
$$Z(f)\le q+1+2I\left(P_f\right)\sqrt{q}\le q+1+2I\left(P_6^{(3)}\right)\sqrt{q}=q+1+0=q+1.$$

So over $\mathbb{F}_q$ when $q>3$, such polynomials
can't give the codewords with weight $(q-1)^2-2(q-1)$. The claim has been shown.

As we have seen, the key point in the argument above is to find out
the distinct families of reducible polynomials which evaluate to
give the codewords with weight $(q-1)^2-2(q-1)$. Table \ref{table-2}
lists the key information for $n_1\left(C_{P_6^{(i)}}\right)$, $3\leq
i\leq8$, with $q\geq 23$.

Let us explain where the lower bound of $q$ comes from briefly in the case $C_{P_6^{(4)}}$. Actually, for $q\ge23$, we claim that there are exactly
5${q-1}\choose{2}$$(q-1)$ such codewords in $C_{P_6^{(4)}}$. Any
other such codewords could only come from evaluating a linear
combination of $\{1, x, x^2, x^3, y, y^{-1}\}$ in which  both $\{x,x^2,x^3\}$ and
$\{y,y^{-1}\}$ appears with at least one element having nonzero coefficients (since otherwise we are in
a case previously covered). Such polynomial will be absolutely irreducible. The maximal polygon of such polynomials associates to the polynomials with $x^3,y,y^{-1}$ having nonzero coefficients, as $f=a_1+a_2x+a_3x^2+a_4x^3+a_5y+a_6y^{-1}$ where $a_4, a_5,a_6\neq 0$. By Proposition $\ref{pr}$, the number of zeros  of $f$ in the torus $(\mathbb{F}_q^*)^2$ has a bound:
$$Z(f)\le q+1+2I(P_f)\sqrt{q}=1+q+4\sqrt{q}.$$
When $q\ge23$, $Z(f)<2q-2$.
Thus such polynomial can never have $2q-2$ zeros when $q\geq23$. Any other smaller polygons have fewer interior points and then have lower upper bound. So all such polynomials can never have $2q-2$ zeros when $q\ge23$. The lower bound of $q\geq 23$ is also valid for other $C_{P_6^{(i)}}$, $5\leq i\leq8$. 

Due to Table \ref{table-2}, we still have the following three
cases to verify:
\begin{enumerate}
	\item When $q\ge 23$, any two codes of $C_{P_6^{(3)}}$, $C_{P_6^{(5)}}$ and
$C_{P_6^{(6)}}$ are pairwise monomially inequivalent;
	\item Over $\mathbb{F}_q$, where $3\nmid(q-1)$, $C_{P_6^{(7)}}$ and
any one codes of $C_{P_6^{(3)}}$, $C_{P_6^{(5)}}$, $C_{P_6^{(6)}}$
are pairwise monomially inequivalent;	
	\item Over $\mathbb{F}_q$, where $q=2^n$, $n\in\mathbb{Z}_+$,
$C_{P_6^{(4)}}$ and $C_{P_6^{(8)}}$ are pairwise monomially inequivalent.
\end{enumerate}

It is worth to mention that we could make sure $n_2\left(C_{P_6^{(4)}}\right)=0$ when $q=2^n$ and $n\ge 5$; and we could settle down the value of $n_2\left(C_{P_6^{(i)}}\right)$ for $i=6,7,8$ and $q\ge 25$. Thus, it is sufficient to check that no two enumerator polynomials of each
codes over $\mathbb{F}_q$, $q\leq 23$, in Table \ref{table-a.2} are exactly the same, which guarantees the codes are pairwise inequivalent in each cases above. The way
to find out $n_2\left(C_{P_6^{(i)}}\right)$, $3\leq i\leq8$, are similar to
that of $n_1\left(C_{P_6^{(3)}}\right)$ before. So we just list the key
information in Table \ref{table-3}, say the distinct families of reducible polynomials
which evaluate to give the codewords with weight $(q-1)^2-(2q-3)$,
of each codes as before. The verifications are left to interested
readers.

We have reached our conclusion, due to Table \ref{table-3}.
\end{proof}

Similarly, we could give a complete classification of monomial equivalence
class of $C_{P_6^{(i)}}$, $9\leq i\leq13$ as in Proposition
\ref{n_1}.
\begin{proposition}
    For $q\geq9$, no two codes of $C_{P_6^{(i)}}$, $9\leq i\leq13$, are monomially equivalent.
\end{proposition}
\begin{proof}[Sketch of the proof.] We classify the toric codes of
$C_{P_6^{(i)}}$, $9\leq i\leq13$, by $n_1\left(C_{P_6^{(i)}}\right)$, see Table \ref{table-4}.

Based on the information in Table \ref{table-4}, there are three cases left to be
determined by $n_2\left(C_{P_6^{(i)}}\right)$:
\begin{enumerate}
	\item $C_{P_6^{(9)}}$ and $C_{P_6^{(11)}}$ are monomially inequivalent;
	\item $C_{P_6^{(10)}}$ and $C_{P_6^{(13)}}$ are monomially inequivalent;
	\item Over $\mathbb{F}_q$, where $q=2^m$, $m\in\mathbb{Z}_+$,
$C_{P_6^{(12)}}$ and any one of $C_{P_6^{(9)}}$, $C_{P_6^{(11)}}$
are monomially inequivalent.
\end{enumerate}

We collect the key information about $n_2\left(C_{P_6^{(i)}}\right)$ in Table \ref{table-5}. Our conclusion has been reached by Table \ref{table-5}.\end{proof}

\appendix 
\section*{Acknowledgements}

The Magma  Program and GAP Program is provided by John Little, to whom we
are extremely thankful, especially for his patience and helpful
discussions. The work of X. Luo is supported by Beihang University start-up fund (Grant No. YWF-14-RSC-026). S. S.-T. Yau thanks the start-up fund from Tsinghua University. And H. Zuo gratefully acknowledges the support of NSFC (Grant No. 11401335) and the start-up fund from Tsinghua University.

\bigskip

\renewcommand{\thetable}{A.\arabic{table}}
\section*{Tables of the Enumerator Polynomials for Toric Codes}
\setcounter{table}{0}

In this appendix, we list all the tables mentioned in our proof of
Theorem \ref{th2}, where the weight enumerator polynomials are
defined as follows:
$$
W_C(x)=\sum_{i=0}^{(q-1)^2} A_ix^i
$$
where $A_i=|\{w\in C:wt(w)=i\}|$, for the $k=6$ toric codes. All the
polynomials are computed by using GAP code with guava package and toric package from [Joy].
\begin{table}
\caption{}\label{table-a.1}
\begin{tabular}{cll}
\hline Over $\mathbb{F}_{7}$&${P_6}^{(1)}$:&$1+36x^{6}+540x^{12}+4320x^{18}+\cdots$\\
& ${P_6}^{(2)}$:&$1+90x^{12}+600x^{18}+2790x^{24}+\cdots$\\
& ${P_6}^{(3)}$:&$1+120x^{18}+360x^{24}+5832x^{25}+\cdots$\\
& ${P_6}^{(4)}$:&$1+120x^{18}+810x^{24}+2160x^{26}+\cdots$\\
& ${P_6}^{(5)}$:&$1+120x^{18}+360x^{24}+648x^{25}+\cdots$\\
& ${P_6}^{(6)}$:&$1+120x^{18}+360x^{24}+648x^{25}+\cdots$\\
& ${P_6}^{(7)}$:&$1+120x^{18}+576x^{24}+216x^{25}+\cdots$\\
& ${P_6}^{(8)}$:&$1+120x^{18}+774x^{24}+2376x^{26}+\cdots$\\
& ${P_6}^{(9)}$:&$1+180x^{24}+1080x^{25}+2916x^{26}+\cdots$\\
& ${P_6}^{(10)}$:&$1+270x^{24}+432x^{25}+4212x^{26}+\cdots$\\
& ${P_6}^{(11)}$:&$1+180x^{24}+1080x^{25}+2700x^{26}+\cdots$\\
& ${P_6}^{(12)}$:&$1+288x^{24}+1728x^{25}+2484x^{26}+\cdots$\\
& ${P_6}^{(13)}$:&$1+270x^{24}+1296x^{25}+4860x^{26}+\cdots$\\
& ${P_6}^{(14)}$:&$1+540x^{20}+180x^{24}+1944x^{25}+\cdots$\\
\hline

Over $\mathbb{F}_{8}$&${P_6}^{(1)}$:&$1+147x^{14}+1470x^{21}+10535x^{28}+\cdots$\\
& ${P_6}^{(2)}$:&$1+245x^{21}+1225x^{28}+558x^{35}+\cdots$\\
& ${P_6}^{(3)}$:&$1+245x^{28}+588x^{35}+11662x^{36}+\cdots$\\
& ${P_6}^{(4)}$:&$1+245x^{28}+735x^{35}+1029x^{36}+\cdots$\\
& ${P_6}^{(5)}$:&$1+245x^{28}+735x^{35}+1029x^{36}+\cdots$\\
& ${P_6}^{(6)}$:&$1+245x^{28}+588x^{35}+686x^{36}+\cdots$\\
& ${P_6}^{(7)}$:&$1+245x^{28}+588x^{35}+1715x^{36}+\cdots$\\
& ${P_6}^{(8)}$:&$1+245x^{28}+735x^{35}+343x^{36}+\cdots$\\
& ${P_6}^{(9)}$:&$1+294x^{35}+1715x^{36}+4459x^{37}+\cdots$\\
& ${P_6}^{(10)}$:&$1+49x^{28}+441x^{35}+2058x^{37}+\cdots$\\
& ${P_6}^{(11)}$:&$1+294x^{35}+2058x^{36}+4116x^{37}+\cdots$\\
& ${P_6}^{(12)}$:&$1+294x^{35}+3430x^{36}+4116x^{37}+\cdots$\\
& ${P_6}^{(13)}$:&$1+441x^{35}+2058x^{36}+9261x^{37}+\cdots$\\
& ${P_6}^{(14)}$:&$1+1029x^{30}+294x^{35}+3087x^{36}+\cdots$\\
\hline
Over $\mathbb{F}_{9}$&${P_6}^{(1)}$:&$1+448x^{24}+3360x^{32}+22848x^{40}+\cdots$\\
& ${P_6}^{(2)}$:&$1+560x^{32}+2240x^{40}+10304x^{48}+\cdots$\\
& ${P_6}^{(3)}$:&$1+448x^{40}+896x^{48}+21504x^{49}+\cdots$\\
& ${P_6}^{(4)}$:&$1+448x^{40}+1888x^{48}+2048x^{50}+\cdots$\\
& ${P_6}^{(5)}$:&$1+448x^{40}+1408x^{48}+1536x^{49}+\cdots$\\
& ${P_6}^{(6)}$:&$1+448x^{40}+896x^{48}+2048x^{49}+\cdots$\\
& ${P_6}^{(7)}$:&$1+448x^{40}+1408x^{48}+1024x^{49}+\cdots$\\
& ${P_6}^{(8)}$:&$1+448x^{40}+1376x^{48}+4864x^{50}+\cdots$\\
& ${P_6}^{(9)}$:&$1+448x^{48}+2516x^{49}+7168x^{50}+\cdots$\\
& ${P_6}^{(10)}$:&$1+2208x^{48}+8704x^{51}+17280x^{52}+\cdots$\\
& ${P_6}^{(11)}$:&$1+448x^{48}+2408x^{49}+4608x^{50}+\cdots$\\
& ${P_6}^{(12)}$:&$1+704x^{48}+4608x^{49}+7936x^{50}+\cdots$\\
& ${P_6}^{(13)}$:&$1+672x^{48}+3072x^{49}+16128x^{50}+\cdots$\\
& ${P_6}^{(14)}$:&$1+1792x^{42}+448x^{48}+4608x^{49}+\cdots$\\
\hline
\end{tabular}
\end{table}

\begin{table}
\caption{}\label{table-a.2}
\begin{tabular}{cll}
 \hline
Over $\mathbb{F}_{11}$&${P_6}^{(3)}$:&$1+1200x^{70}+1800x^{80}+61000x^{81}+\cdots$\\
&${P_6}^{(4)}$:&$1+1200x^{70}+2250x^{80}+2000x^{82}+\cdots$\\
&${P_6}^{(5)}$:&$1+1200x^{70}+2000x^{80}+4000x^{82}+\cdots$\\
& ${P_6}^{(6)}$:&$1+1200x^{70}+2000x^{80}+1000x^{81}+\cdots$\\
& ${P_6}^{(7)}$:&$1+1200x^{70}+1800x^{80}+2000x^{81}+\cdots$\\
& ${P_6}^{(8)}$:&$1+1200x^{70}+2750x^{80}+2500x^{82}+\cdots$\\
& ${P_6}^{(9)}$:&$1+900x^{80}+5000x^{81}+18000x^{82}+\cdots$\\
& ${P_6}^{(10)}$:&$1+1350x^{80}+3000x^{81}+3000x^{82}+\cdots$\\
& ${P_6}^{(11)}$:&$1+1100x^{80}+3000x^{81}+11000x^{82}+\cdots$\\
& ${P_6}^{(12)}$:&$1+1400x^{80}+10000x^{81}+19500x^{82}+\cdots$\\
& ${P_6}^{(13)}$:&$1+1350x^{80}+6000x^{81}+40500x^{82}+\cdots$\\
\hline
Over $\mathbb{F}_{13}$&${P_6}^{(3)}$:&$1+2640x^{108}+3168x^{120}+145152x^{121}+\cdots$\\
& ${P_6}^{(4)}$:&$1+2640x^{108}+4248x^{120}+3456x^{122}+\cdots$\\
& ${P_6}^{(5)}$:&$1+2640x^{108}+3168x^{120}+1728x^{122}+\cdots$\\
& ${P_6}^{(6)}$:&$1+2640x^{108}+3168x^{120}+1728x^{121}+\cdots$\\
& ${P_6}^{(7)}$:&$1+2640x^{108}+4320x^{120}+1728x^{121}+\cdots$\\
& ${P_6}^{(8)}$:&$1+2640x^{108}+4824x^{120}+864x^{122}+\cdots$\\
& ${P_6}^{(9)}$:&$1+1584x^{120}+8640x^{121}+38016x^{122}+\cdots$\\
& ${P_6}^{(10)}$:&$1+2376x^{120}+5184x^{123}+32832x^{124}+\cdots$\\
& ${P_6}^{(11)}$:&$1+1584x^{120}+5184x^{121}+19008x^{122}+\cdots$\\
& ${P_6}^{(12)}$:&$1+2448x^{120}+19008x^{121}+40608x^{122}+\cdots$\\
& ${P_6}^{(13)}$:&$1+2376x^{120}+10368x^{121}+85536x^{122}+\cdots$\\
\hline
Over $\mathbb{F}_{16}$&${P_6}^{(3)}$:&$1+6825x^{180}+6300x^{195}+425250x^{196}+\cdots$\\
& ${P_6}^{(4)}$:&$1+6825x^{180}+7875x^{195}+3375x^{197}+\cdots$\\
& ${P_6}^{(5)}$:&$1+6825x^{180}+6975x^{195}+13500x^{199}+\cdots$\\
& ${P_6}^{(6)}$:&$1+6825x^{180}+6975x^{195}+3375x^{196}+\cdots$\\
& ${P_6}^{(7)}$:&$1+6825x^{180}+8550x^{195}+3375x^{196}+\cdots$\\
& ${P_6}^{(8)}$:&$1+6825x^{180}+7875x^{195}+3375x^{196}+\cdots$\\
& ${P_6}^{(9)}$:&$1+3150x^{195}+16875x^{196}+94500x^{197}+\cdots$\\
& ${P_6}^{(10)}$:&$1+4725x^{195}+40500x^{198}+20250x^{199}+\cdots$\\
& ${P_6}^{(11)}$:&$1+3825x^{195}+10125x^{196}+47250x^{197}+\cdots$\\
& ${P_6}^{(12)}$:&$1+3150x^{195}+47250x^{196}+94500x^{197}+\cdots$\\
& ${P_6}^{(13)}$:&$1+4725x^{195}+20250x^{196}+212625x^{197}+\cdots$\\
\hline
Over $\mathbb{F}_{17}$&${P_6}^{(3)}$:&$1+8960x^{208}+7680x^{224}+581632x^{225}+\cdots$\\
& ${P_6}^{(4)}$:&$1+8960x^{208}+9600x^{224}+4096x^{228}+\cdots$\\
& ${P_6}^{(5)}$:&$1+8960x^{208}+7680x^{224}+8192x^{228}+\cdots$\\
& ${P_6}^{(6)}$:&$1+8960x^{208}+7680x^{224}+4096x^{225}+\cdots$\\
& ${P_6}^{(7)}$:&$1+8960x^{208}+7680x^{224}+8192x^{225}+\cdots$\\
& ${P_6}^{(8)}$:&$1+8960x^{208}+11648x^{224}+2048x^{226}+\cdots$\\
& ${P_6}^{(9)}$:&$1+3840x^{224}+20480x^{225}+122880x^{226}+\cdots$\\
& ${P_6}^{(10)}$:&$1+5760x^{224}+12288x^{227}+24576x^{229}+\cdots$\\
& ${P_6}^{(11)}$:&$1+3840x^{224}+12288x^{225}+61440x^{226}+\cdots$\\
& ${P_6}^{(12)}$:&$1+5888x^{224}+53248x^{225}+129024x^{226}+\cdots$\\
& ${P_6}^{(13)}$:&$1+5760x^{224}+24576x^{225}+276480x^{226}+\cdots$\\
\hline
\end{tabular}
\end{table}
\begin{table}
\begin{tabular}{cll}
\hline
Over $\mathbb{F}_{19}$&${P_6}^{(3)}$:&$1+14688x^{270}+11016x^{288}+1032264x^{289}+\cdots$\\
& ${P_6}^{(4)}$:&$1+14688x^{270}+13770x^{288}+5832x^{292}+\cdots$\\
& ${P_6}^{(5)}$:&$1+14688x^{270}+11016x^{288}+5832x^{292}+\cdots$\\
& ${P_6}^{(6)}$:&$1+14688x^{270}+11016x^{288}+5832x^{289}+\cdots$\\
& ${P_6}^{(7)}$:&$1+14688x^{270}+14904x^{288}+5832x^{289}+\cdots$\\
& ${P_6}^{(8)}$:&$1+14688x^{270}+16686x^{288}+2916x^{290}+\cdots$\\
& ${P_6}^{(9)}$:&$1+5508x^{288}+29160x^{289}+198288x^{290}+\cdots$\\
& ${P_6}^{(10)}$:&$1+8262x^{288}+17496x^{290}+5832x^{292}+\cdots$\\
& ${P_6}^{(11)}$:&$1+5508x^{288}+17496x^{289}+99144x^{290}+\cdots$\\
& ${P_6}^{(12)}$:&$1+8424x^{288}+81648x^{289}+207036x^{290}+\cdots$\\
& ${P_6}^{(13)}$:&$1+8262x^{288}+34992x^{289}+4416148x^{290}+\cdots$\\
\hline
Over $\mathbb{F}_{23}$&${P_6}^{(3)}$:&$1+33880x^{418}+20328x^{440}+2757832x^{441}+\cdots$\\
& ${P_6}^{(4)}$:&$ 1+33880x^{418}+25410x^{440}+53240x^{448}+\cdots$\\
& ${P_6}^{(5)}$:&$ 1+33880x^{418}+20328x^{440}+10648x^{445}+\cdots$\\
& ${P_6}^{(6)}$:&$ 1+33880x^{418}+20328x^{440}+10648x^{441}+\cdots$\\
& ${P_6}^{(7)}$:&$1+33880x^{418}+20328x^{440}+21296x^{441}+\cdots$\\
& ${P_6}^{(8)}$:&$ 1+33880x^{418}+30734x^{440}+5324x^{442}+\cdots$\\
& ${P_6}^{(9)}$:&$1+10164x^{440}+53240x^{441}+447216x^{442}+\cdots$\\
& ${P_6}^{(10)}$:&$1+15246x^{440}+31944x^{446}+63888x^{447}+\cdots$\\
& ${P_6}^{(11)}$:&$1+10164x^{440}+31944x^{441}+223608x^{442}+\cdots$\\
& ${P_6}^{(12)}$:&$1+15488x^{440}+170368x^{441}+463188x^{442}+\cdots$\\
& ${P_6}^{(13)}$:&$1+15246x^{440}+63888x^{441}+1006236x^{442}+\cdots$\\
\hline
\end{tabular}
\end{table}

\bigskip

\begin{table}
\caption{}\label{table-a.3}
\begin{tabular}{cll}
\hline
Over $\mathbb{F}_{25}$&$P_6^{(3)}$:&$1+48576x^{504}+26496x^{528}+4230144x^{529}+\cdots$\\
&$P_6^{(5)}$:&$1+48576x^{504}+26496x^{528}+13824x^{529}+69120x^{536}+\cdots$\\
&$P_6^{(6)}$:&$1+48576x^{504}+26496x^{528}+13824x^{529}+13824x^{536}+\cdots$\\
&$P_6^{(10)}$:&$1+19872x^{528}+13824x^{529}+41472x^{536}+\cdots$\\
&$P_6^{(13)}$:&$1+19872x^{528}+82944x^{529}+1430784x^{530}+\cdots$\\
\hline

Over $\mathbb{F}_{27}$&$P_6^{(5)}$:&$1+67600x^{598}+33800x^{624}+70304x^{633}+\cdots$\\
&$P_6^{(10)}$:&$1+25350x^{624}+52728x^{630}+105456x^{633}+\cdots$\\
\hline

Over $\mathbb{F}_{29}$&$P_6^{(5)}$:&$1+91728x^{700}+42336x^{728}+21952x^{737}+\cdots$\\
&$P_6^{(10)}$:&$1+31752x^{728}+131712x^{736}+65856x^{738}+\cdots$\\
\hline

Over $\mathbb{F}_{31}$&$P_6^{(5)}$:&$1+121800x^{810}+52200x^{840}+27000x^{849}+\cdots$\\
&$P_6^{(10)}$:&$1+39150x^{840}+27000x^{847}+81000x^{852}+\cdots$\\
\hline

Over $\mathbb{F}_{32}$&$P_6^{(5)}$:&$1+139345x^{868}+57660x^{899}+59582x^{910}+\cdots$\\
&$P_6^{(10)}$:&$1+43245x^{899}+893730x^{913}+1787460x^{914}+\cdots$\\
\hline

Over $\mathbb{F}_{37}$&$P_6^{(5)}$:&$1+257040x^{1188}+90720x^{1224}+46656x^{1236}+\cdots$\\
&$P_6^{(10)}$:&$1+68040x^{1224}+139968x^{1237}+279936x^{1238}+\cdots$\\
\hline

Over $\mathbb{F}_{41}$&$P_6^{(5)}$:&$1+395200x^{1480}+124800x^{1520}+320000x^{1536}+\cdots$\\
&$P_6^{(10)}$:&$1+93600x^{1520}+192000x^{1535}+64000x^{1537}+\cdots$\\
\hline

Over $\mathbb{F}_{43}$&$P_6^{(5)}$:&$1+482160x^{1638}+144648x^{1680}+10584x^{1694}+\cdots$\\
&$P_6^{(10)}$:&$1+108486x^{1680}+10584x^{1687}+296352x^{1699}+\cdots$\\
\hline
\end{tabular}
\end{table}

\end{document}